\documentclass[a4paper]{article}

\usepackage{amsmath}
\usepackage{amsfonts,amssymb,amsthm}
\usepackage{mathtools}
\usepackage[english]{babel}
\usepackage[breaklinks=true]{hyperref}
\usepackage[latin1]{inputenc}
\usepackage[T1]{fontenc}
\usepackage{geometry}                
\usepackage{smartref}
\usepackage{needspace}
\usepackage{enumerate}
\usepackage{array}
\usepackage{appendix}
\usepackage{graphicx}

\graphicspath{{Imagenes}{Images}}

\usepackage{fancyhdr}
\usepackage[verbose]{wrapfig}
\usepackage[width=.9\linewidth,font=small,labelfont=bf,labelsep=endash]{caption}
\usepackage[usenames,dvipsnames]{color}
\usepackage{fancybox}
\usepackage{smartref}
\addtoreflist{section}

\allowdisplaybreaks

    \newcommand{\corurl}{red}
    \newcommand{\corcite}{Green}
    \newcommand{\corlink}{blue}
    \newcommand{\corfile}{black}

\hypersetup{linktocpage,colorlinks,urlcolor=\corurl,citecolor=\corcite,linkcolor=\corlink,filecolor=\corfile,pdfnewwindow}

\newcommand{\quot}[1]{\fracdiag{#1}{\sim\hspace*{0.4ex}}}
\newcommand{\quotf}[1]{\fracdiag{#1}{\sim_f\hspace*{0.4ex}}}

\makeatletter
 \newcommand{\raisemath}[1]{\mathpalette{\raisem@th{#1}}}
 \newcommand{\raisem@th}[3]{\raisebox{#1}{$#2#3$}}
\makeatother
\newcommand{\peqsub}[2]{#1_{\raisemath{-0.75pt}{\scriptscriptstyle #2}}}
\newcommand{\fracdiag}[2]{\mathchoice{{\left.\raisebox{.2em}{$#1$}\middle/\raisebox{-.2em}{$#2$}\right.}}{#1/\hspace*{-0.75ex}#2}{#1/\hspace*{-0.75ex}#2}{#1/\hspace*{-0.75ex}#2}}

\newcommand{\dist}{-1.3ex}

\newcommand{\Z}{\mathbb{Z}}     
\newcommand{\R}{\mathbb{R}}     
\renewcommand{\S}{\mathbb{S}}   
\newcommand{\D}{\mathbb{D}}     

\newcommand{\refconchap}[1]{\relax\hyperref[#1]{\textcolor{\corlink}{\sectionref{#1}.}\ref*{#1}}}

\def\QED{{\boldmath$\rule{0.5em}{0.5em}$}}                                
\def\markatright#1{\leavevmode\unskip\nobreak\quad\hspace*{\fill}{#1}}    
\def\qed{\markatright{\QED}}

\newtheorem{theorem}{Theorem}[section]
\newtheorem{definition}[theorem]{Definition}
\newtheorem*{definitions}{Definitions}
\newtheorem{proposition}[theorem]{Proposition}
\newtheorem*{properties}{Properties}
\newtheorem{remark}[theorem]{Remark}

\newtheorem*{remarks}{Remarks}
\newtheorem{lemma}[theorem]{Lemma}

\renewenvironment{proof}[1][Proof]{ \textbf{#1.} }{\qed\\}     

\begin{document}

      \pagestyle{fancy}        \fancyhf{}
      \renewcommand{\headrulewidth}{0pt}
      \fancyfoot[C]{- \thepage\ -}
      \setlength{\headheight}{12.2pt}

      \vspace*{-1ex}

      \centerline{\textsf{\textbf{\Huge{Topology of the Misner Space}}}}

      \mbox{}\vspace*{-0.2ex}

      \centerline{\textsf{\textbf{\Huge{and its $g$-boundary}}}}

      \mbox{}\vspace*{2.5ex}

\centerline{
\begin{tabular}{ccc}
    \textsf{\textbf{Juan Margalef--Bentabol}}${}^{1,2}$       &     \quad & \textsf{\textbf{Eduardo J. S. Villaseñor}}${}^{2}$\\[0ex]
    \href{mailto:juanmargalef@ucm.es}{juanmargalef@ucm.es} & \quad & \href{mailto:ejsanche@math.uc3m.es}{ejsanche@math.uc3m.es}
\end{tabular}}

\vspace*{2ex}

\centerline{ \scriptsize
\begin{tabular}{ll}
    ${}^1$ & Instituto de Estructura de la Materia, CSIC, Serrano 123, 28006 Madrid, Spain.\\
    ${}^2$ & Unidad Asociada al IEM-CSIC, Grupo de Teorías de Campos y Física Estadística,\\
           & Instituto Universitario Gregorio Millán Barbany, Grupo de Modelización y\\
           & Simulación Numérica, Universidad Carlos III de Madrid, Madrid, Spain
\end{tabular}}

\abstract{\centerline{
\begin{minipage}{0.9\linewidth}\setlength{\parindent}{0pt}
     The Misner space is a simplified 2-dimensional model of the 4-dimensional Taub-NUT space that reproduces some of its pathological behaviours. In this paper we provide an explicit base of the topology of the complete Misner space $\R^2_1/boost$. Besides we prove that some parts of this space, that behave like topological boundaries, are equivalent to the $g$-boundaries of the Misner space.
\end{minipage}}}

\setlength{\parindent}{0pt}

  \section{Introduction}
    The Taub-NUT space-time is a spatially homogenous vacuum solution to the Einstein equations that displays many strange behaviours. In order to understand some of its pathologies C.W.~Misner introduced in a seminar entitled \emph{Taub-NUT space as a counterexample to almost anything} \cite{Taub-nut_contraejemplo}, a simpler $2$-dimensional model, the so-called \textbf{Misner space}. This space-time has still some strange behaviours that have been carefully studied in \cite{Haw_y_ellis,thorne1993misner}, and besides, it has been used as a toy-model for different purposes like big bounce models \cite{durin2006closed,hikida2005d}.\vspace*{2ex}

     The geometrical and topological properties of Misner space can be derived by constructing it in terms of the quotient space $\R^2_1/boost$. Following  this approach, we do not only recover the Misner space but also prove that some subsets of it behave like the $g$-boundaries introduced by S.W.~Hawking \cite{Sing_geom_Haw_1966} and R.~Geroch \cite{geroch1968local}. It is important to point out that the abstract definition of the $g$-boundary of a semiriemannian manifold is highly non trivial  and, in practice,  it is very difficult to compute it from the definition. In fact, in \cite{geroch1968local} and \cite{hajicek1970embedding} some examples are obtained but no computation is provided. This construction has an even  worse problem: whenever a boundary construction verifies, as the $g$-boundary does, some reasonable conditions, then a space-time can be constructed with unphysical boundary \cite{geroch1982singular}. Even the fact of being defined using geodesics is not satisfactory owing to an example due to R.~Geroch \cite{geroch1968singularity} of a geodesically complete space-time containing curves of finite length and bounded acceleration (hence there exist physical time-like observers that are incomplete). In part due to  these problems, the study of this and similar constructions has been essentially abandoned. We  think, however, that it is still interesting to verify that in the case of the complete Misner space, the $g$-boundary is what one expects to be.  Besides, for some of these completions, such as for the $g$-boundary, there exists a canonical way of defining a minimal refining of the topology such that the points of the boundary become T2-separated from all the rest of the points \cite{Hausdorff_separability} removing some of the worst problems of these constructions; hence, they  may regain some interest in the recent future.\vspace*{2ex}

    The paper is structured as follows, in section \ref{section hyperbolic rotation} we carefully study the Misner space as a quotient space, providing an explicit base of the quotient topology and deducing from it some topological properties (and problems).  In section \ref{section further considerations} we discuss, in some detail, the behaviour of the light-like geodesics in the extended Misner space. The concept of the $g$-boundary is introduced in section \ref{g-boudnary} and apply it to the the  Misner space. We show that the topology defined by the $g$-boundary coincides, precisely, with the quotient topology obtained in section \ref{section hyperbolic rotation}. We end this paper in section \ref{conclusions} with the discussion of the main results and our conclusions. For the convenience of the reader, we also provide an appendix where the necessary mathematical background is presented.

   \section{The Misner space as a quotient space}\label{section hyperbolic rotation}
        In this section, we construct the Misner space $M$  by considering the quotient space of the Minkowski space-time $\R^2_1=(\R^2,-dT^2+dX^2)$  under a discrete group $G_{\peqsub{\theta}{0}}$ of Lorentz  boosts. In this approach,  $M$ arises as the quotient space $\{T>|X|\}/G_{\peqsub{\theta}{0}}$ (or, equivalently $\{-T>|X|\}/G_{\peqsub{\theta}{0}}$), which will allow us to understand the pathological behaviour of its geodesics. Moreover, by considering the quotient over two adjacent quadrants of $\R^2_1$ plus the line in between them (e.g $\{T<X\}/G_{\peqsub{\theta}{0}}$), we get an isometric extension $C$ of $M$, that we will refer to as the \emph{extended Misner space}. This extension \emph{partially} solves the geodesic incompleteness problem of $M$. However, as we will see, it is not possible to extend $C$ further to completely solve the the geodesic incompleteness.

    \subsection{General facts about the Misner space}
        The Lorentz group $O_{1,1}(\R)$ is formed by the linear maps $H:\R^2\rightarrow\R^2$,  which preserve the Minkowski product $\langle \peqsub{P}{1},\peqsub{P}{2}\rangle=-\peqsub{T}{1}\peqsub{T}{2}+\peqsub{X}{1}\peqsub{X}{2}$. The subgroup composed by the space-orientation preserving and the time-orientation preserving maps is called \textbf{special orthochronous Lorentz group}:\vspace*{\dist}

        \[SO^+_{1,1}(\R)=\left\{\peqsub{H}{\theta}=\begin{pmatrix}
            \cosh\theta & \sinh\theta\\
            \sinh\theta & \cosh\theta
        \end{pmatrix}\ /\ \theta\in\R\right\}=\left\{\peqsub{H}{\theta}=\begin{pmatrix}
                \peqsub{\gamma}{\theta}  & \peqsub{v}{\theta}\peqsub{\gamma}{\theta}\\
                \peqsub{v}{\theta}\peqsub{\gamma}{\theta} & \peqsub{\gamma}{\theta}
            \end{pmatrix}\ /\ \theta\in\R\right\}\]

        where $\peqsub{v}{\theta}=\tanh\theta$ is the velocity of one inertial frame with respect to another one, $\theta$ can be interpreted as the rapidity \cite{levy1980speed} and, as usual, $\peqsub{\gamma}{\theta}=\frac{1}{\sqrt{1-\peqsub{v}{\theta}^2}}$. The boosts matrices $\peqsub{H}{\theta}$ satisfy the following:

        \begin{properties}\mbox{}\renewcommand{\labelenumi}{P.\arabic{enumi}}
            \begin{enumerate}
               \item \begin{minipage}[t]{\linewidth}  As $\peqsub{H}{\theta}$ is linear, it maps lines (through the origin) to lines (through the origin),\begin{wrapfigure}[11]{r}{0.32\linewidth}\vspace*{-3ex}
                        \centerline{\hspace*{2ex}\includegraphics[page=2,width=1.1\linewidth]{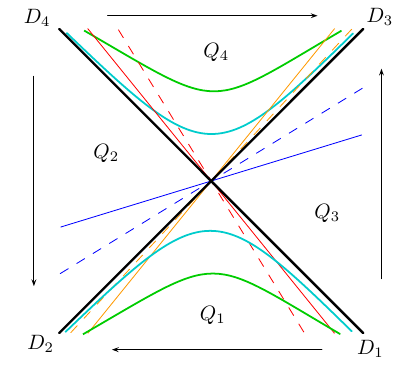}}
                    \end{wrapfigure} and it has as invariant lines the diagonals $\{X=\pm T\}$. Indeed each semi-line is invariant.
           \item  $H_{\theta}(L^i_{\psi})=L^i_{\psi+\theta}$ where for a fixed given quadrant $Q_i$, we denote $L^i_{\psi}$ the open semi-line in $Q_i$ that begins at the origin with hyperbolic angle $\psi$. From now on, $Q_i$ is the corresponding open quadrant (in between the diagonals) and $D_i$ is the corresponding open semi-diagonal where the origin is excluded.\label{propiedades p1 p2 p3}
           \item  $H_{\theta}(l^\pm_{\peqsub{T}{0}})=l^\pm_{\peqsub{T}{0}e^{\mp\theta}}$ where $l^{\pm}_{\peqsub{T}{0}}$ is the line with slope $\pm 1$ that cuts the semi-diagonal at the point $(T,X)=(\peqsub{T}{0},\mp \peqsub{T}{0})$.
           \item For $r\geq0$, the sets \vspace*{\dist}

           \[\mathcal{H}^\pm_{r}=\{(T,X)\in\R^2\ /\ -T^2+X^2=\pm r^2\}\qquad\qquad\]
            are invariant. Notice that for $r>0$, the hyperbolas $\mathcal{H}^-_{r}$ are space-like and $\mathcal{H}^+_{r}$ are time-like.
                    \end{minipage}
            \end{enumerate}
        \end{properties}

 \centerline{\includegraphics[width=0.75\linewidth]{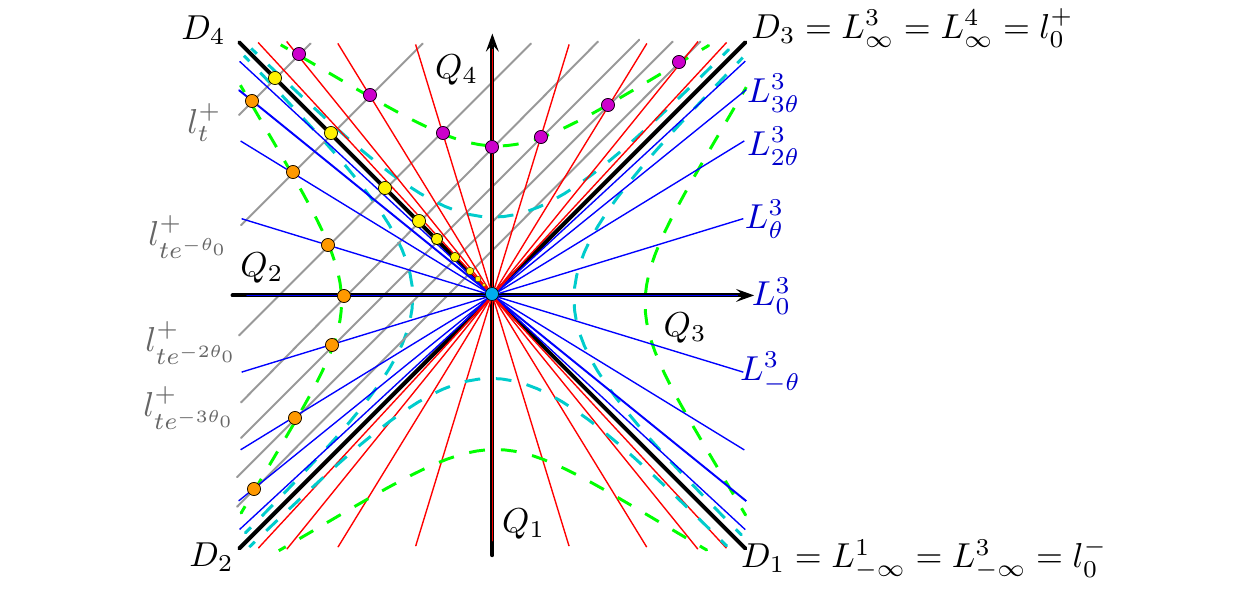}}

     If $\peqsub{\theta}{0}>0$ then $G_{\peqsub{\theta}{0}}=\langle H_{\peqsub{\theta}{0}}\rangle$ is an infinite cyclic group, and for every $P=(T,X)\neq (0,0)$, the orbit $G_{\peqsub{\theta}{0}}(P)$ is an infinite countable set. This action over $D_i$  is properly discontinuous and free. In particular, if $\peqsub{P}{0}=(\peqsub{T}{0},\pm \peqsub{T}{0})\in D_i$, then:\vspace*{\dist}

              \[ \peqsub{P}{n}=(H_{\peqsub{\theta}{0}})^n(P)=H_{n\peqsub{\theta}{0}}(P)=(\peqsub{T}{0}e^{\pm n\peqsub{\theta}{0}},\pm \peqsub{T}{0}e^{\pm n\peqsub{\theta}{0}})\]

     Thus, for forward iterations, the points over the diagonal $\{T=-X\}$ accumulate over the origin exponentially while the points over the diagonal $\{T=X\}$ grow to infinity exponentially. By using the exponential map $\mathrm{exp}:\R\rightarrow\R^+$, $x\mapsto e^{x}=|\peqsub{T}{0}|$, it is clear that each of the quotient spaces $\widetilde{D}_i=D_i/G_{\peqsub{\theta}{0}}$ is homeomorphic to a circle $\S^1\cong \R/\mathrm{mod}\,\peqsub{\theta}{0}$.\vspace*{2ex}

     The action of $G_{\peqsub{\theta}{0}}$ over $Q_i$  is also properly discontinuous and free. In particular $M^-_{\peqsub{\theta}{0}}=Q_1/G_{\peqsub{\theta}{0}}$ (resp. $M^+_{\peqsub{\theta}{0}}=Q_4/G_{\peqsub{\theta}{0}}$) is a smooth Hausdorff manifold homeomorphic to the cone $\R^-\times \S^1$ (resp. $ \R^+\times\S^1$).  In fact, by expressing  the Minkowski metric in hyperbolic coordinates  $(T,X)=(t \cosh(x),t\sinh(x))$ the action of $G_{\peqsub{\theta}{0}}$ is given by $H_{\peqsub{\theta}{0}}(t,x)=(t,x+\peqsub{\theta}{0})$ and the metric in $M^\pm_{\peqsub{\theta}{0}}$ is:\vspace*{\dist}

     \[dm_\pm^2=-dt^2+t^2dx^2\,,\quad (t,x)\in \R^\pm\times (\R/\mathrm{mod}\,\peqsub{\theta}{0})\]

     In particular, the circles $\{t=\peqsub{t}{0}\}\subset M^\pm_{\peqsub{\theta}{0}}$ are space-like. Notice that for $\peqsub{\theta}{0}=\pi$, both $M^\pm_\pi$ are isometric to the Misner space:\vspace*{\dist}

     \[dm^2=-\frac{dt'^2}{t'}+t'd\theta'^2\,,\quad (t',\theta')\in \R^+\times (\R/\mathrm{mod}\,2\pi)\]

     The future directed light-like geodesics over the Misner space $M^-_{\peqsub{\theta}{0}}$ are:\vspace*{\dist}

     \[\peqsub{\gamma}{\pm}(\tau)=\biggl(t(\tau),x(\tau)\biggr)=\biggr(-|\peqsub{t}{0}|\sqrt{1-2|\peqsub{\dot x}{0}|\tau},\peqsub{x}{0}\mp\frac{1}{2}
                     \ln\left(\rule{0ex}{2.5ex}1-2|\peqsub{\dot x}{0}|\tau\right)\biggl)\]

     where we have replaced $\peqsub{t}{0}$ by $-|\peqsub{t}{0}|$ in order to emphasize that $t$ ranges over $\R^-=(-\infty,0)$. The geodesics satisfy $\dot{t}>0$, $\peqsub{\dot t}{0}=|\peqsub{\dot x}{0}||\peqsub{t}{0}|$, and $\mathrm{sign}(\dot{x})=\pm1$. Summarizing, we have that the two future directed light-like geodesics obtained are incomplete as $\tau\in(-\infty,\frac{1}{2|\peqsub{\dot x}{0}|})$, but notice that, as $x$ is increasing and grows to infinite, they turn around the cylinder infinitely many times. We can see graphically the behaviour of the geodesics $\peqsub{\gamma}{\pm}$ in figure \ref{figure geod espiral}, the closer to $\S^1\times\{0\}$ (which corresponds to the apex, and hence it does not belong to $M^-_{\peqsub{\theta}{0}}$), the quicker with respect to the affine parameter a light-like geodesic turns around the cylinder.\label{warped product cone}\medskip

    \begin{figure}[h!]
        \centering
        \includegraphics[page=1,width=0.8\linewidth,page=1]{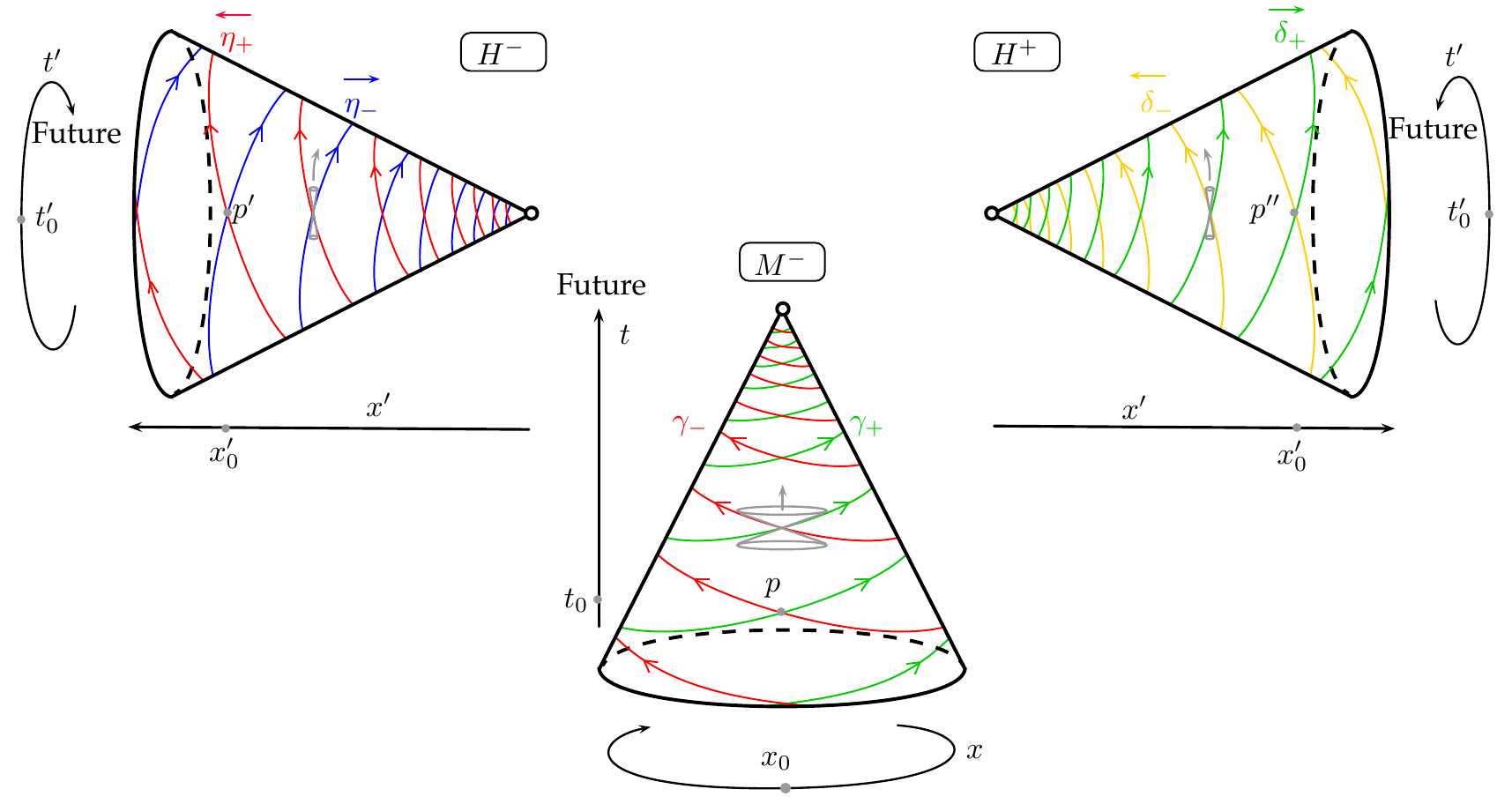}
        \caption{All the light-like geodesics turn around the semi-cylinder infinitely many times. The apex is in the future for both $\gamma_\pm$ and for $\delta_-$ and $\eta_+$ while it is in the past for $\delta_+$ and $\eta_-$.}\label{figure geod espiral}
    \end{figure}

    Similarly,  $H^-_{\peqsub{\theta}{0}}=Q_2/G_{\peqsub{\theta}{0}}$ (resp. $H^+_{\peqsub{\theta}{0}}=Q_3/G_{\peqsub{\theta}{0}}$) is a smooth Hausdorff manifold homeomorphic to the cone $\S^1\times\R^\pm$. In hyperbolic coordinates $(T,X)=(x' \sinh(t'),x'\cosh(t'))$, the metric in $H^\pm_{\peqsub{\theta}{0}}$ is\vspace*{\dist}

    \[dh_\pm^2=-x'^2 dt'^2+d x'^2\,,\quad (t',x')\in (\R/\mathrm{mod}\peqsub{\theta}{0})\times \R^\pm\]

    In this case, the circles $\{x'=\peqsub{x}{0}\}$ are time-like. The future directed light-like geodesics over $H^+_{\peqsub{\theta}{0}}$ are:\vspace*{\dist}

    \[\peqsub{\delta}{\pm}(\tau)=\biggl(t'(\tau),x'(\tau)\biggr)=\biggr(\peqsub{t}{0}\pm\frac{1}{2}\ln\left(\rule{0ex}{2.5ex}1\pm2\peqsub{\dot{t}}{0}\tau\right),\peqsub{x}{0}\sqrt{1\pm2\peqsub{\dot{t}}{0}\tau}\biggl)\qquad\left|\begin{array}{ll}
        \tau\in\left(-\tfrac{1}{2\peqsub{\dot{t}}{0}},+\infty\right) &\ \text{ for }\ \delta_+\\[1ex]
        \tau\in\left(-\infty,\tfrac{1}{2\peqsub{\dot{t}}{0}}\right) & \ \text{ for }\ \delta_-
    \end{array}\right.\]

    Notice that now $t$ is the periodic coordinate and we have $\dot{t}>0$, and $\pm\peqsub{\dot{x}}{0}=\peqsub{x}{0}\peqsub{\dot{t}}{0}>0$.  The light-like geodesics $\peqsub{\eta}{\pm}$ of $H^-$ would have the same expressions (we could replace $\peqsub{x}{0}$ by $-|\peqsub{x}{0}|$) with $\peqsub{\eta}{-}$ aiming also towards the apex.\vspace*{2ex}

    If we now consider the quotient over the whole plane, $Y_{\peqsub{\theta}{0}}=\R^2/G_{\peqsub{\theta}{0}}$, we do not obtain a smooth manifold as the origin is a fixed point, furthermore some other problems arise as we will see in section \ref{problems with the extension}. However the action over one of the semi-planes $S^\pm=\{T<\pm X\}$ is free and properly discontinuous, so the quotient is a smooth Hausdorff manifold which topologically is homeomorphic to two cones without the apex (one for each quadrant) and a circle corresponding to the semi-diagonal. The cones can be deformed into semi-cylinders and glued along the circle to obtain a whole cylinder. In particular, the metric in the \textbf{extended Misner space}  $C^+_{\peqsub{\theta}{0}}=S^+/G_{\peqsub{\theta}{0}}\cong \R \times \S^1$ can be written~as:\vspace*{\dist}

    \[ds_+^2=-z d\vartheta_+^2-2d\vartheta_+ dz\,,\quad (z,\vartheta_+)\in \R \times (\R/\mathrm{mod}\,2 \peqsub{\theta}{0})\]

    where the coordinates $(z,\vartheta_+)$ are related to those on $M^-_{\peqsub{\theta}{0}}$ and $H^+_{\peqsub{\theta}{0}}$ through\vspace*{\dist}

    \[ (z,\vartheta_+)=\left\{ \begin{array}{lcl} \left(-\frac{t^2}{4},2x+2\ln(-t)\right) & \textrm{ if  }&(t,x)\in M^-_{\peqsub{\theta}{0}}\\[1.5ex]
     \left(\frac{x'^2}{4},-2t'+2\ln(x')\right)&\textrm{ if  }&(t',x')\in H^+_{\peqsub{\theta}{0}} \end{array}\right. \]

    Notice that the light-like circle $\{z=0\}\subset C^+$ corresponds to $D_1/G_{\peqsub{\theta}{0}}$. On the other hand, the metric in the \textbf{extended Misner space}   $C^-_{\peqsub{\theta}{0}}=S^-/G_{\peqsub{\theta}{0}}\cong \R \times \S^1$ is:\vspace*{\dist}

    \[ds_-^2=-z d\vartheta_-^2+2d\vartheta_- dz\,,\quad (z,\vartheta_-)\in \R \times (\R/\mathrm{mod}\,2 \peqsub{\theta}{0})\]

    where\vspace*{\dist}

    \[ (z,\vartheta_-)=\left\{ \begin{array}{lcl} \left(-\frac{t^2}{4},2x-2\ln(-t)\right) &\textrm{ if  }&(t,x)\in M^-_{\peqsub{\theta}{0}} \\[1.5ex] \left(\frac{x'^2}{4},-2t'-2\ln(-x')\right)&\textrm{ if  }&(t',x')\in H^-_{\peqsub{\theta}{0}} \end{array}\right. \]

    In this case the light-like circle $\{z=0\}\subset C^-$ corresponds to $D_2/G_{\peqsub{\theta}{0}}$.\vspace*{2ex}

    From now on, let us focus on $C^+$. The light-like geodesics $\peqsub{\gamma}{\pm}$ can be expressed in the $\{z,\peqsub{\vartheta}{+}\}$~-~coordinates to give the future directed ($\dot{z}>0$) light-like geodesics $\peqsub{\gamma}{\pm}$ of $M^-$ passing through $\peqsub{z}{0}=-|\peqsub{z}{0}|<0$:\vspace*{\dist}

    \[\begin{array}{l}
        \widetilde{\gamma}_+(\tau)=\biggl(z(\tau),\vartheta_+(\tau)\biggr)=\biggr(\peqsub{\dot{z}}{0}\tau-|\peqsub{z}{0}|,\peqsub{\vartheta}{0}\biggr)\\[2ex]
        \widetilde{\gamma}_-(\tau)=\biggl(z(\tau),\vartheta_+(\tau)\biggr)
               =\biggr(\peqsub{\dot{z}}{0}\tau-|\peqsub{z}{0}|,\peqsub{\vartheta}{0}+2\ln\left(1-\frac{\peqsub{\dot{z}}{0}}{|\peqsub{z}{0}|}\tau\right)\biggr)
    \end{array}\qquad\tau\in\left(-\infty,\frac{|\peqsub{z}{0}|}{\peqsub{\dot{z}}{0}}\right)\]

    Similarly, the future directed light-like geodesics $\peqsub{\delta}{\pm}$ of $H^+$ passing through $\peqsub{z}{0}>0$ are:\vspace*{\dist}

    \[\begin{array}{lcl}
       \widetilde{\delta}_+(\tau)=\biggl(z(\tau),\vartheta_+(\tau)\biggr)=\biggr(\peqsub{\dot{z}}{0}\tau+\peqsub{z}{0},\peqsub{\vartheta}{0}\biggr) & \quad & \tau\in\left(-\dfrac{\peqsub{z}{0}}{\peqsub{\dot{z}}{0}},\infty\right)\\[2ex]
       \widetilde{\delta}_-(\tau)=\biggl(z(\tau),\vartheta_+(\tau)\biggr)=\biggr(\peqsub{z}{0}-|\peqsub{\dot{z}}{0}|\tau,\peqsub{\vartheta}{0}
       +2\ln\left(1-\frac{|\peqsub{\dot{z}}{0}|}{\peqsub{z}{0}}\tau\right)\biggr) & \quad & \tau\in\left(-\infty,\dfrac{\peqsub{z}{0}}{\peqsub{|\dot{z}}{0}|}\right)
    \end{array}\]

    where notice that now $sgn(\dot{z})=\pm1$. Clearly $\widetilde{\gamma}_+$ and $\widetilde{\delta}_+$ can be extended to $\tau\in\R$, in fact, they can be glued together through its limit point at the $z=0$ level. However the geodesics $\widetilde{\gamma}_-$ and $\widetilde{\delta}_-$ cannot be extended and remain incomplete. Notice that $\dot\vartheta_+\leq$ for all the geodesics.\vspace*{2ex}

    We see that we have partially solved our ``pathology'' as we have managed to ``unwrap'' $\widetilde{\gamma}_+$ and $\widetilde{\delta}_+$, but we have ``wrapped twice'' $\widetilde{\gamma}_-$ and $\widetilde{\delta}_-$. The $\peqsub{z}{0}=0$ case has to be studied separately:

    \begin{remark}\mbox{}\label{remark aumenta velocidad}\\
      The future directed light-like geodesics with $\peqsub{z}{0}=0$ are\quad $\left\{\begin{array}{ll}
       \widetilde{\rho}_+(\tau)= \biggr(\peqsub{\dot{z}}{0}\tau,\peqsub{\vartheta}{0}\biggr) & \tau\in\R    \\[2ex]
       \widetilde{\rho}_-(\tau)= \biggr(0,\peqsub{\vartheta}{0}+2\ln\left(1-\frac{\peqsub{\dot\vartheta}{0}}{2}\tau\right)\biggr)     & \tau\in\left(-\infty,\frac{2}{\peqsub{\dot\vartheta}{0}}\right)
    \end{array}\right.$\vspace*{2ex}

    $\widetilde{\rho}_+$ is exactly the complete $\widetilde{\gamma}_+$ and $\widetilde{\delta}_+$ glued through their limit point, while $\widetilde{\rho}_-$ is incomplete, remains always in the subset $\S^1\times\{0\}$ and verifies $\dot\vartheta_-<0$ (as well as $\widetilde{\gamma}_-$ and $\widetilde{\delta}_-$). Furthermore it verifies a quite astonishing property. For a given $\peqsub{\tau}{0}$ if we define a $\peqsub{\tau}{k}$ for every $k\in\Z$ such that:\vspace*{\dist}

    \[\left(1-\frac{\peqsub{\dot\vartheta}{0}}{2}\peqsub{\tau}{0}\right)=\left(1-\frac{\peqsub{\dot\vartheta}{0}}{2}\peqsub{\tau}{k}\right)e^{k\peqsub{\theta}{0}}\]

    then we have that for every $k\in \Z$:\vspace*{\dist}

    \[\widetilde{\rho}_+(\peqsub{\tau}{0})=\widetilde{\rho}_+(\peqsub{\tau}{k})\qquad\qquad \dot{\widetilde{\rho}}_+(\peqsub{\tau}{0})\neq\dot{\widetilde{\rho}}_+(\peqsub{\tau}{k})\]

    Hence the curve passes infinitely many times through the same point $\widetilde{\rho}_+(\peqsub{\tau}{0})$ but with ``longer'' tangent vector (but notice that its modulus remains constant and equal to zero!).
    \end{remark}

    For the $\{t,\vartheta_-\}$ extended coordinates we have analogous results, although now the roles of each pair of geodesics are exchanged, and of course, we have to consider $H^-$ and $\peqsub{\eta}{\pm}$ instead of $H^+$ and $\peqsub{\delta}{\pm}$ (see figure \ref{figure_cilindros}). Besides, this construction can be used to obtain an incomplete compact spacetime \cite[page 77]{JAVALOYES_SANCHEZ_lorentz} by gluing a positive slice $\{z=cte_+\}$ and a negative one $\{z=cte_-\}$ with a suitable deformation, obtaining a torus with the $\{z=0\}$ slice in it.\vspace*{2ex}

    We see then that we have two inequivalent inextensible extensions of $M^-$, however they both turn out to be again incomplete. Once at that point we should try to understand why this happens and if it is possible to find an even better coordinate system to solve the pathology completely. Notice that in order to extend the space-time $M^-$, we have followed the incomplete geodesics and we have defined a new coordinate to extend them, so one might wonder if we can do the same over the extensions $C^\pm$ and find where the wrapped geodesics go. It is not hard to believe that they go to the upper cone (isometric to the lower one under the mapping $t\mapsto-t$, see figure in section \ref{section further considerations}) but if we apply the same technique to solve the double wrapped geodesics, we will find that the already ``straight'' ones, become wrapped again. Actually, as we mentioned before, no further extension exists and here is where topology will turn out to be essential to understand why the pathological behaviour appears in the first place and why we can partially, but not completely, get rid off it.

                         \begin{figure}[h!]
     \centering
     \includegraphics[width=0.7\linewidth]{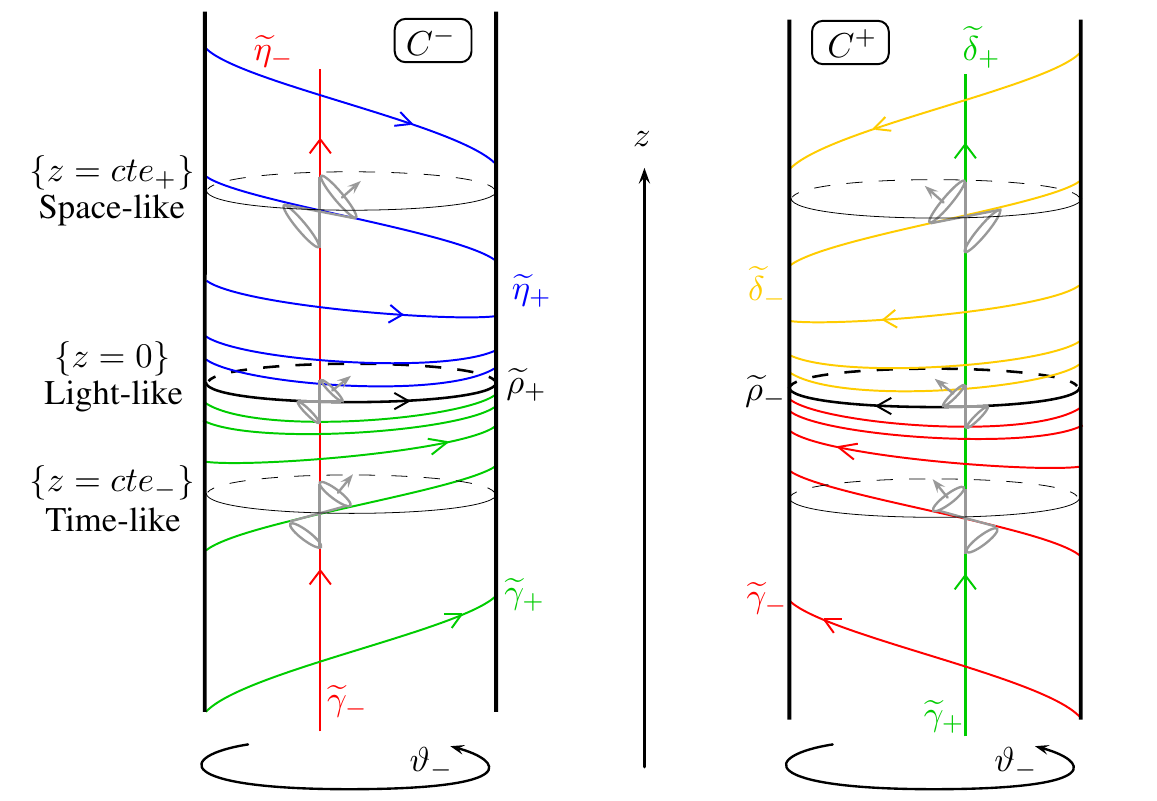}
     \caption{$C^-$ is the result of gluing the lower cone $M^-$ and the left one $H^-$ while $C^+$ is obtained by gluing $M^-$ and $H^+$. A similar picture can be found on page 172 of \cite{Haw_y_ellis} but notice that there, the cones involved are the lateral and the upper ones, hence the time-like and space-like sections are exchanged.\vspace*{-2ex}}\label{figure_cilindros}
    \end{figure}

    \subsection{Topology induced by the discrete hyperbolic rotation}{\label{Topology}}

    We will now provide a local base for the topology of the complete Misner space $\R^2_1/boost$. To this end we are going to consider some local base of $\R^2$ and saturate their open sets to obtain a base of the quotient topology (see \ref{recipe of quotient topology}).

       \begin{itemize}
          \item[\ovalbox{$P\in Q_i$}] In local hyperbolic coordinates $P\in Q_i\subset \R^2$ can be written as $(\peqsub{r}{0},\peqsub{\psi}{0})$. Let us define\vspace*{\dist}

                \[B_1(P,\varepsilon)=\left\{(r,\psi)\in Q_i\ :            r\in(\peqsub{r}{0}-\varepsilon,\peqsub{r}{0}+\varepsilon)\quad
                \psi\in(\peqsub{\psi}{0}-\varepsilon,\peqsub{\psi}{0}+\varepsilon)\right\}\]

            Using $P.1$-$P.4$, we have\vspace*{\dist}

            \[H_{n\peqsub{\theta}{0}}\left(\rule{0ex}{2.2ex}B_1(P,\varepsilon)\right)
            =B_1\left(\rule{0ex}{2.2ex}H_{n\peqsub{\theta}{0}}(P),\varepsilon\right)\]

            In order to saturate, we have to join all these possible images (see remark \ref{remark saturate f}):\vspace*{\dist}

            \[S_{\sim}\left[\rule{0ex}{2.2ex}B_1(P,\varepsilon)\right]=\bigcup_{n\in\Z} B_1\left(\rule{0ex}{2.2ex}H_{n\peqsub{\theta}{0}}(P),\varepsilon\right)\]

            Finally we consider the local base $\beta(P)=\{B_1(P,\varepsilon)\}_{\varepsilon\in(0,\peqsub{\varepsilon}{P})}$ for every $P\in Q_i$ with $\peqsub{\varepsilon}{P}>0$ small enough such that $B_1(P,\varepsilon)$ does not contain two related points (hence the union is disjoint). If we now look at the glued space, we have simply an open ``squared'' ball over the cone.

             \centerline{\includegraphics[width=1\linewidth]{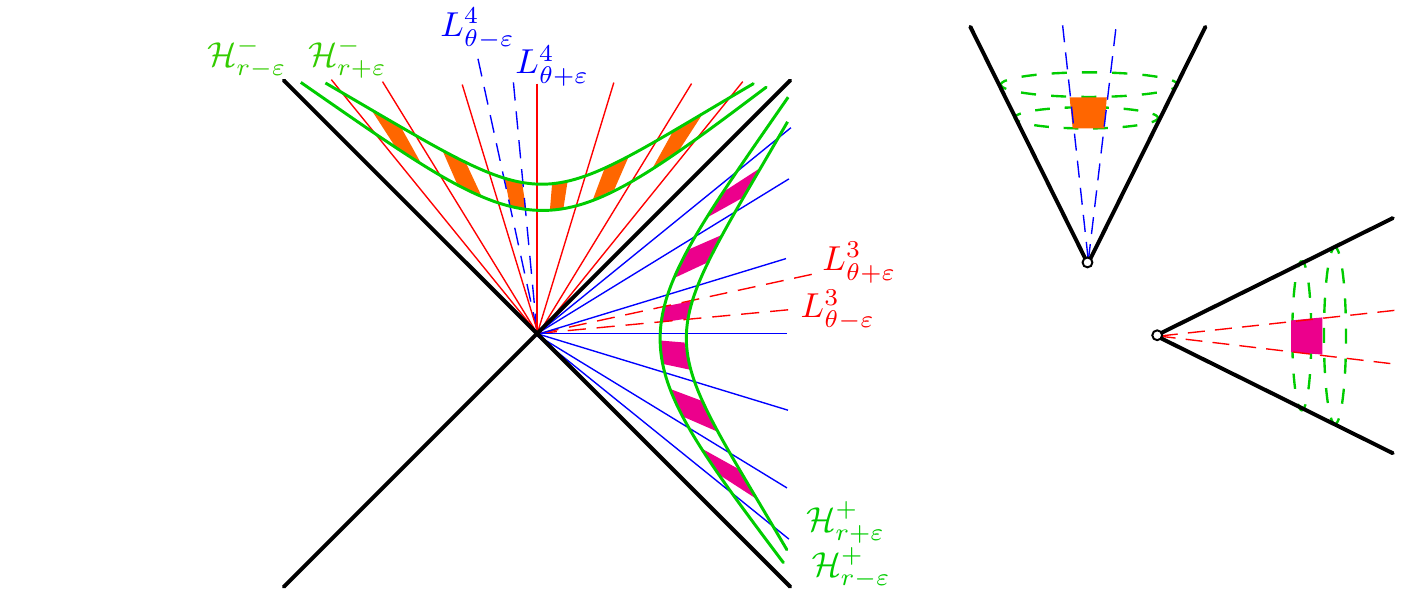}}

          \item[\ovalbox{$P\in D_i$}] Let $P=(\peqsub{T}{0},\peqsub{T}{0})\in D_3$ (analogously for the remaining $D_i$ just changing some signs), we define\vspace*{\dist}

              \[B_2(P,\varepsilon,\varepsilon')=
              \left(\bigcup_{\delta\in(-\varepsilon,\varepsilon)}l^-_{\peqsub{T}{0}+\delta}\right)
              \bigcap\left(\left[\bigcup_{r\in[0,\varepsilon')}\mathcal{H}^+_r\right]\bigcup
              \left[\bigcup_{r\in[0,\varepsilon')}\mathcal{H}^-_r\right]\right)\]

              Using $P.1$-$P.4$ we have $H_{n\peqsub{\theta}{0}}\left(\rule{0ex}{2.2ex}B_2(P,\varepsilon,\varepsilon')\right)
              =B_2\left(\rule{0ex}{2.2ex}H_{n\peqsub{\theta}{0}}(P),\varepsilon e^{n\peqsub{\theta}{0}},\varepsilon'\right)$ and then

              \[\hspace*{-3ex}
              S_{\sim}\left[\rule{0ex}{2.2ex}B_1(P,\varepsilon,\varepsilon')\right]=\bigcup_{n\in\Z} B_1\left(\rule{0ex}{2.2ex}H_{n\peqsub{\theta}{0}}(P),\varepsilon e^{n\peqsub{\theta}{0}},\varepsilon'\right)\]

              Finally we take the local base $\beta(P)=\{B_2(P,\varepsilon,\varepsilon')\}_{\varepsilon\in(0,\peqsub{\varepsilon}{P}),\varepsilon'>0}$ with $\peqsub{\varepsilon}{P}$ small enough such that $B_2(P,\varepsilon,\varepsilon')$ does not contain the origin and does not contain two related points.\vspace*{2ex}

              As the ball $B=B_2(P,\varepsilon,\varepsilon')$ is a bunch of (finite piece of) light-like geodesics each one having a piece over $Q_4$, a piece over $Q_2$, and a point over the diagonal $D_3$, hence over the glued space we have a bunch of geodesics each having a piece over $M^+$, a piece over $H^+$ and a point over $\widetilde{D}_3$. Therefore given a point $q\in \widetilde{D}_i$, any basic open neighbourhood $U$ over the glued space is an open interval over $\widetilde{D}_i$, together with two ``thickened'' geodesics (similar to $\peqsub{\widetilde{\gamma}}{\pm}$), each one turning infinitely many times around the corresponding adjacent cone, where the direction of rotation is determined by the evolution of the geodesics over the universal cover as we explain in section~\ref{section further considerations}.

              \centerline{\includegraphics[clip,trim=0.5cm 0.5cm 0.2cm -0.25cm,width=0.8\linewidth]{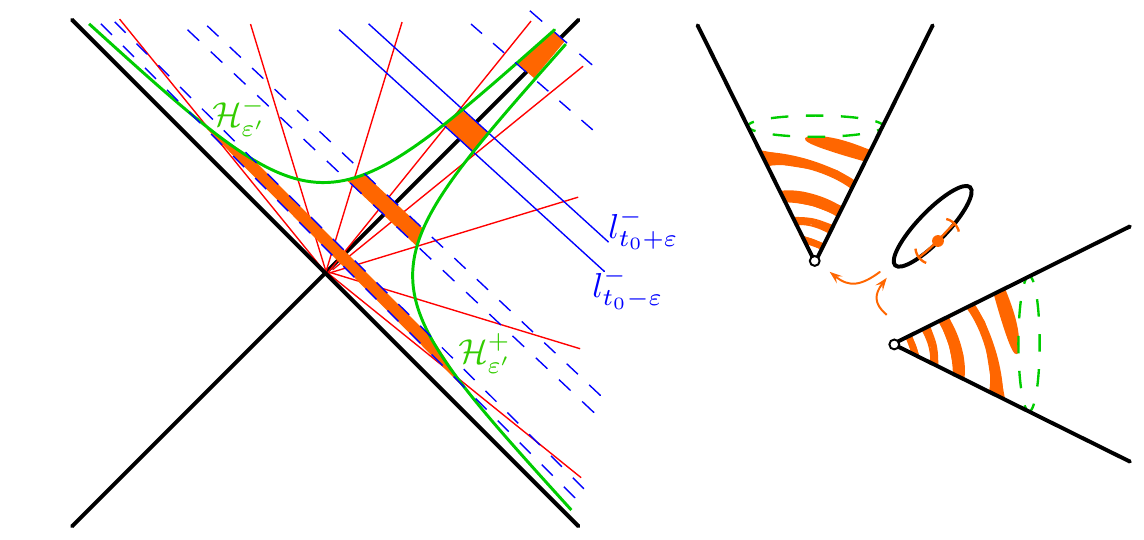}}

          \item[\ovalbox{$P=0$}] We define\vspace*{\dist}

          \[B_3(0,\varepsilon)=\left(\bigcup_{\delta\in(-\varepsilon,\varepsilon)}l^-_{\delta}\right)\bigcap\left(\bigcup_{\delta\in(-\varepsilon,\varepsilon)}l^+_{\delta}\right)\bigcap\left(\left[\bigcup_{r\in[0,\varepsilon e^{-\peqsub{\theta}{0}})}\mathcal{H}^+_r\right]\bigcup\left[\bigcup_{r\in[0,\varepsilon e^{-\peqsub{\theta}{0}})}\mathcal{H}^-_r\right]\right)\]

           Using again $P.1$-$P.4$ we obtain:\vspace*{\dist}

           \[H_{n\peqsub{\theta}{0}}\left(\rule{0ex}{2.2ex}B_3(0,\varepsilon)\right)\cup
             H_{-n\peqsub{\theta}{0}}\left(\rule{0ex}{2.2ex}B_3(0,\varepsilon)\right)=
             B_3\left(\rule{0ex}{2.2ex}0,\varepsilon e^{n\peqsub{\theta}{0}}\right)\]

           Saturating it, we obtain the region in between the four branches of the hyperbolas $\mathcal{H}^+_{\varepsilon e^{-\peqsub{\theta}{0}}}$ and $\mathcal{H}^-_{\varepsilon e^{-\peqsub{\theta}{0}}}$ (pushing the lines to infinity):\vspace*{\dist}

           \[S_{\sim}\left[\rule{0ex}{2.2ex}B_3(0,\varepsilon)\right]=\left[\bigcup_{r\in[0,\varepsilon e^{-\peqsub{\theta}{0}})}\mathcal{H}^+_r\right]\bigcup\left[\bigcup_{r\in[0,\varepsilon e^{-\peqsub{\theta}{0}})}\mathcal{H}^-_r\right]\]

           Each branch of the hyperbola correspond to a ``straight'' circle  over each cone ($\{t=cte\}$ in the Misner space). Taking all the hyperbolas with $r\in(0,\varepsilon e^{-\peqsub{\theta}{0}})$ gives simply an annulus around the apex on each cone. The $r=0$ case correspond to the diagonals, which in the glued space are the four circles plus the conic center. Hence a basic open neighbourhood of the origin is the point itself, the four circles and four annulus of ``length'' $\varepsilon e^{-\peqsub{\theta}{0}}$ on each cone.\vspace*{1.5ex}

           \centerline{\includegraphics[clip,trim=1.2cm 0.86cm 0.2cm 0.9cm,width=0.9\linewidth]{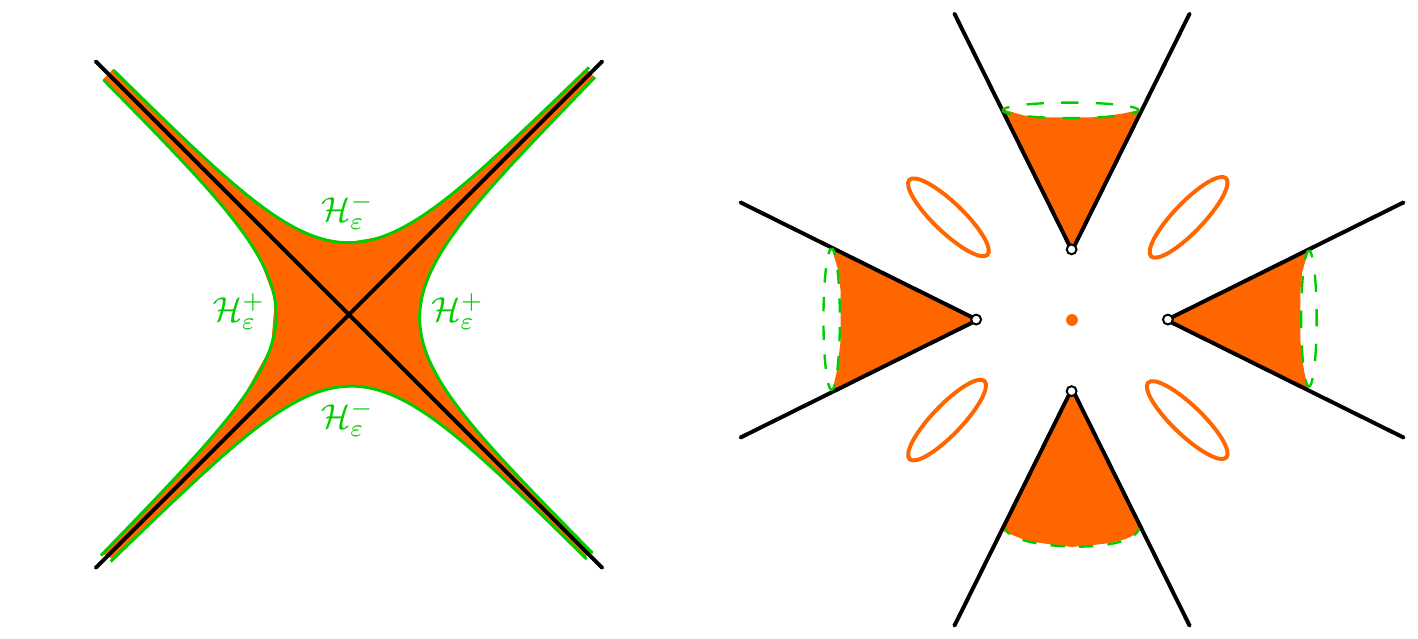}}
       \end{itemize}

       \subsection{Topological properties of the extended quotient}\label{topological properties}

          As every point has a countable base, the space $Y_{\peqsub{\theta}{0}}=\R^2/G_{\peqsub{\theta}{0}}$ is a first-countable space, indeed we can consider just the points with rational coordinates and we obtain that the space is in fact second-countable. When we restrict the action to one of the semi-planes $S^\pm=\{T<\pm X\}$ we obtain a smooth Hausdorff manifold, but whenever we extend the action including two adjacent semi-diagonals, we have that it is no longer Hausdorff, as can be seen taking (saturated) basic open sets $U\in\beta(x)$ and $V\in\beta(y)$ with $x\in D_4$ and $y\in D_3$. If we consider two points of the same diagonal without the origin, we can choose small enough neighbourhoods such that they do not overlap.

           \centerline{\includegraphics[width=0.8\linewidth,clip,trim=0cm 0cm 1.2cm -0.3cm]{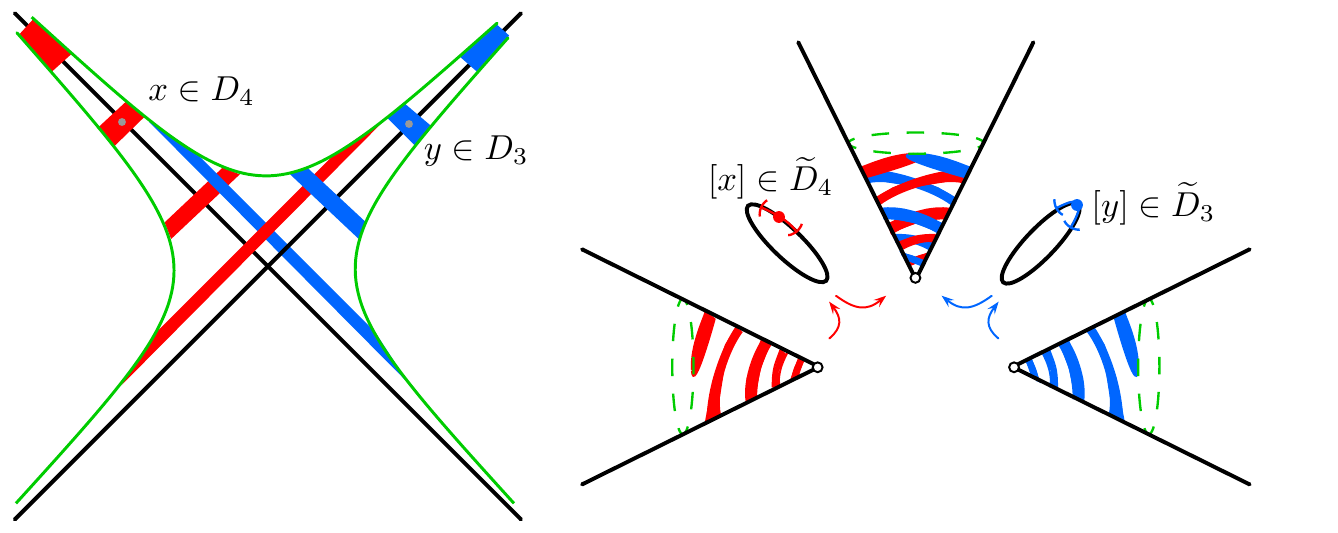}}

          Moreover, as we see on the last image of the previous section, the origin cannot be separated from any point of the circles $\widetilde{D}_i$ as they are contained in any open neighbourhood of the origin. Let us then study which separation axioms does the quotient spaces verify (see section \ref{subsection hausdorff}):

         \needspace{\baselineskip}
          \begin{proposition}\mbox{}
            \begin{enumerate}
             \item The quotient space $Y_{\peqsub{\theta}{0}}=\R^2/G_{\peqsub{\theta}{0}}$ is T0 but not T1.
             \item The quotient space $Y^*_{\peqsub{\theta}{0}}=(\R^2\setminus\{0\})/G_{\peqsub{\theta}{0}}$ is T1 but not T2.
             \item The quotient spaces $C^\pm_{\peqsub{\theta}{0}}=\{T<\pm X\}/G_{\peqsub{\theta}{0}}$ are T2.
             \item Every $x\in D_i$ (analogously of $\widetilde{D}_i$) is not closed over $Y_{\peqsub{\theta}{0}}$ but it is closed over $Y_{\peqsub{\theta}{0}}^*$.
            \end{enumerate}
          \end{proposition}
          \begin{proof}\mbox{}\\
                The only tricky statement is the last one. First notice that removing the origin we have a T1 space, so we might expect that the origin fails to be closed, but surprisingly it is closed. However for every $x\in \widetilde{D}_i$ we have that $\{x\}^c$ is not open. If we regard the universal cover, we see that the problem lies on the fact that the equivalence class $[x]$ of $x\in D_i$ is a countable set that accumulates over the origin but that does not contain it, so $[x]$ is not closed in $\R^2$ (but it is closed in $\R^2\setminus\{0\}$, where it has no accumulation point). On the other hand, the equivalence class of the origin is just a point $[0]=\{0\}$ which is closed over the plane (this is a general result for quotient topologies).
          \end{proof}

       \subsection{Problems with the extension of the group action}\label{problems with the extension}
          The statements made in the preceding sections allow us to know why the action group does not work nicely over some regions (see definition \refconchap{definition properly discont} on appendix  \ref{subsection actions}). The action over the whole plane without the origin $\R^2\setminus\{0\}$ is free and verifies $PD1$, but it fails to satisfy $PD2$. Here we see perfectly why the apparently weird property $PD2$ is required in order to ensure that the resultant manifold is Hausdorff. If we consider the action over the whole plane, we will not obtain a smooth manifold as the origin is a fixed point of $H_{\peqsub{\theta}{0}}$, in fact the action is neither free nor $PD2$, and hence we will obtain a non-smooth non-Hausdorff manifold (known as non-Hausdorff orbifold).

          \begin{remark}\mbox{}\\
            The continuation over $M^+$ of any geodesic of $M^-$ going through the origin, is not univocally determined as it can be broken at the origin. In fact two different geodesics in $M^-$ may merge into one over $M^+$. Despite this pathology, we have a natural way of assigning the continuation, namely, following the straight line over the universal cover.
          \end{remark}

   \section{Further Considerations about the whole Misner Space}\label{section further considerations}
            Some behaviours of the Misner space can be illustrated in the following figure (compare with fig. \ref{figure_cilindros}).

            \centerline{\includegraphics[width=1\linewidth]{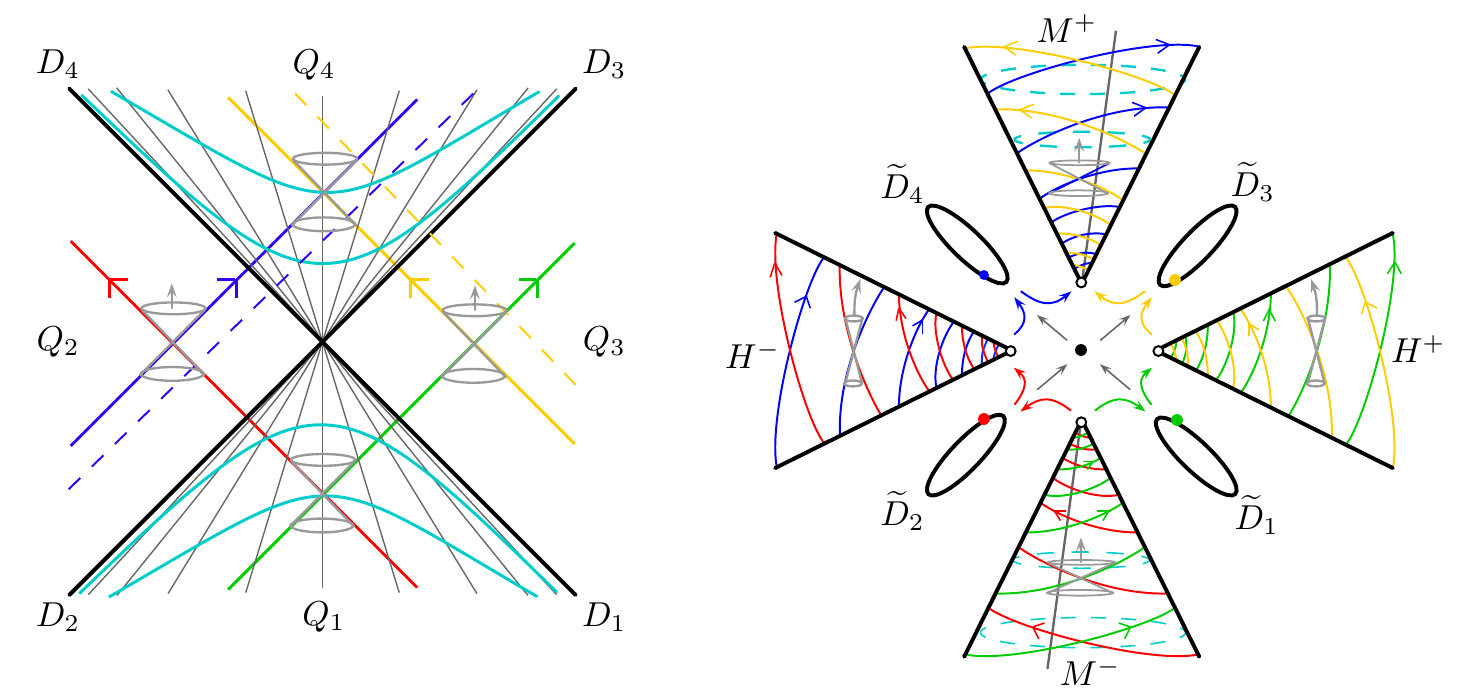}\label{FIG_geodesics_Misner_completo}}

            In this figure we see, for instance, why the character of the coordinates $(z,\peqsub{\vartheta}{\pm})$ is exchanged depending on the sign of the $z$-coordinate in the extended Misner space $C^\pm$. It is also important to notice that the spinning happens in the specific sense it does depending on the direction of the corresponding geodesic over the universal cover.\vspace*{2ex}

            Summarizing, if we reach the apex turning infinitely many times, we will ``jump'' to the adjacent circle (the right one for the clockwise and the left one for the counterclockwise with respect to the origin), touch it at exactly one point and ``jump'' again to the next quadrant, according to the rules described by the arrows shown in the previous figure. The other possibility is reaching the apex without turning infinitely many times, which means that over the universal cover we have not crossed any semi-diagonal, hence the geodesic crosses through the origin. The time-like geodesics through the origin must go from $Q_1$ to $Q_4$. Over the glued space they are world lines that reach the apex of $M^-$ without turning infinitely many times, ``jump'' to the origin and ``jump'' again to $M^+$. The light-like geodesics through the origin are precisely the semi-diagonals $D_i$. Over the glued space these geodesics turn on one of the lower circles $\widetilde{D}_1$ or $\widetilde{D}_2$ infinitely many times but with a finite affine parameter, then ``jump'' to the origin and finally ``jump'' again to the opposite circle as the arrows of the figure suggest. The quotations in the word jump come from the fact that according to the topology, no jump exists (such curves are of the form $\widetilde{\gamma}=p\circ\gamma$, a composition of continuous functions -see def. \ref{def quotient topology}-).\vspace*{2ex}

        It is worth mentioning that quite often this space is not depicted completely right \cite{durin2006closed,hikida2005d,jonsson2005visualizing} as the circles $\widetilde{D}_i$ are missing. Probably this lack of precision is not important for many purposes, but it is of capital importance if we want to understand in detail the Misner space.

         \begin{remark}\mbox{}\\
          A uniformly accelerated observer over the lateral quadrants $Q_2$ or $Q_3$ follows a hyperbola (through translation we may consider that it has the semi-diagonals as its asymptotes). As this kind of hyperbolas over those quadrants are the time-like circles $\S^1\times\{x=cte\}$ over the lateral cones, those observers describe closed time-like curves over the glued space.
         \end{remark}

        \centerline{\includegraphics[width=0.9\linewidth,clip,trim=2.2cm. 0.2cm 0.1cm -0.2cm,page=1]{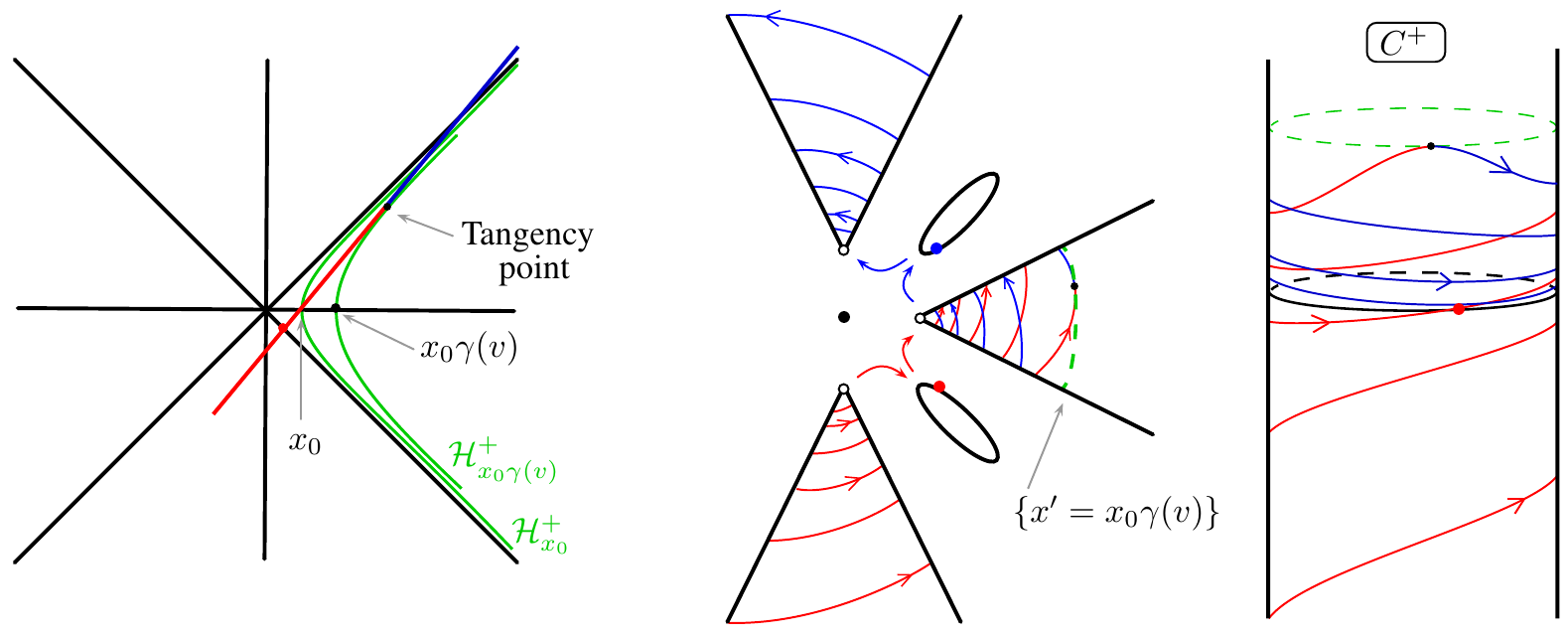}}

         Any inertial observer over the lateral cones will eventually jump to the upper cone in the same way that any observer in the Minkowski space-time that does not pass through the origin will eventually reach the semi-diagonals $\{T=|X|\}$, hence the only way that an observer can remain on the lateral cylinders is by experiencing a perpetual acceleration. Besides notice that all the time-like geodesics that do not cross the origin will be tangent to some hyperbola $\mathcal{H}_r$ over the lateral quadrants (the slope of these hyperbolas tend to $\pm1$). Computing the intersection between the geodesic and a generic hyperbola $\mathcal{H}_r$, and imposing the tangency condition, gives that the maximum hyperbolic radius (and hence the ``lateral height'') attained in the lateral cones is:\vspace*{\dist}

         \[\peqsub{x}{\mathrm{max}}=\frac{\peqsub{x}{0}}{\sqrt{1-v^2}}=\peqsub{x}{0}\gamma(v)\qquad \]

         Over the upper and lower quadrants, it crosses every possible hyperbola and hence over the cones it goes from $t\rightarrow -\infty$ and goes to $t\rightarrow +\infty$ (where the $t$ coordinate is defined separately over each cone). It is interesting also to describe what would happen over the extended space $C^+$: the time-like geodesic will turn finitely many times around the lower semi-cylinder, cross the $\{z=0\}$ level set, turn finitely many times until the maximum of the $z$ coordinate is reached. From this point on, the geodesic falls again towards $\{z=0\}$ but now turning infinitely many times. Notice in particular that the geodesic is incomplete as $C^+$ is formed gluing $M^-$ and $H^+$ through the circle $\widetilde{D}_1$, but not the upper cone $M^+$.

    \section{\texorpdfstring{{\boldmath$g$}}{g}-boundary}\label{g-boudnary}

       The results presented in the previous sections show that some parts of the extended Misner space, namely the circles $\widetilde{D}_i$  and the origin, behave somehow like boundaries of the ``adjacent cones''  in the sense that every incomplete geodesic of the cones has a limit point that is over the circles or in the origin. The behaviour of these sets resembles that of the $g$-boundary introduced by Hawking \cite{Sing_geom_Haw_1966} and Geroch \cite{geroch1968local}. Therefore it is worth to compute the $g$-boundary of $M^- $ and see if, as expected, we recover the two adjacent circles and the origin obtained with the quotient topology.\vspace*{2ex}

      The construction of the  $g$-boundary provides a way to build, and glue properly, a boundary $\partial_g M$ to an incomplete semiriemannian manifold $(M,g)$. The idea is quite tricky as this completion has to be done just making reference to the manifold itself to define something that will be ``outside'' of it. It is important to notice that the same manifold space $M$ can have different $g$-boundaries (end-points of the incomplete geodesics regarded in the ambient space) depending on the metric. For instance consider the unitary disk $\D\subset\R^2$:
      \begin{itemize}
         \item If we use the Euclidean metric, all the geodesics are incomplete and the boundary $\partial_g\D$ is $\S^1$.
         \item If we consider it to be the Poincaré's disk with the hyperbolic metric, then all the geodesics are complete and hence the boundary $\partial_g\D$ is empty.
         \item If we consider the topological disk as the whole sphere without the north pole $\S^2\setminus\{\peqsub{p}{N}\}\subset\R^3$ with the round metric, we see that the boundary $\partial_g(\S^2\setminus\{\peqsub{p}{N}\})$ is just a point. Now $\S^2\setminus\{\peqsub{p}{N}\}$ and the round metric can be pulled-back to the unitary disk through the stereographic projection $M\rightarrow \R^2$ followed by a contraction $\R^2\rightarrow \D$. Hence, with this particular metric, $\partial_g\D$ is just one point.
      \end{itemize}

      Now we proceed with a quick review of Geroch's method to build the $g$-boundary.

      \subsection{Short review of the construction of the \texorpdfstring{{\boldmath$g$}}{g}-boundary}
        Let $(M,g)$ be a semiriemannian manifold and $G=TM\setminus\{0\}=\{(p,v)\in TM\ /\ v\neq0\}$. For any $(p,v)$ there exists a unique maximal geodesic $\peqsub{\gamma}{(p,v)}:I\rightarrow M$ such that $\peqsub{\gamma}{(p,v)}(0)=p$ and $\peqsub{\dot\gamma}{(p,v)}(0)=v$. By geodesic we understand a standard geodesic (not just pregeodesic) defined over $I=[0,b)$ with $b\in\R^+\cup\{\infty\}$. With this notation, the complete geodesics are the ``forward complete'' geodesics. The reason for doing this is that we can then bijectively associate the geodesics with the reduced tangent bundle $G$. We now define the function:\vspace*{\dist}

        \[\varphi:G\longrightarrow \R^+\cup\{\infty\}\]

        such that $\varphi(p,v)$ is the total affine length (in the forward direction) of the corresponding geodesic $\peqsub{\gamma}{(p,v)}$. Clearly $\varphi$ is infinite if and only if the geodesic in question is complete. We now define the sets:\vspace*{\dist}

        \[\peqsub{G}{\!I}=\{x\in G\ /\ \varphi(x)<\infty\}\qquad\quad\begin{array}{l}
          H=G\times\R^+\\
          H_+=\{(p,v,\tau)\in H\ /\ \tau<\varphi(p,v)\}
        \end{array}\]

        $\peqsub{G}{\!I}$ is formed by the incomplete geodesics, $H$ is the set of all possible geodesics and all possible affine parameters, while $H_+$ restricts the possible affine parameters to those ones where the geodesic is well defined. Therefore we have a well defined map:\vspace*{\dist}

        \[\begin{array}{cccc}
           \Psi:&     H_+    & \longrightarrow & M\\
                & (p,v,\tau) &    \longmapsto  & \peqsub{\gamma}{(p,v)}(\tau)
        \end{array}\]

        We now topologize the set $\peqsub{G}{\!I}$. For a given open set $U\subset M$, we define the subset of $\peqsub{G}{\!I}$:\vspace*{\dist}

        \[S(U)=\left\{(p,v)\in \peqsub{G}{\!I}\ /\ \text{there exists } A\in \peqsub{N}{\!H}(p,v,\varphi(p,v))\ \text{with}\ \Psi(A\cap H_+)\subset U\rule{0ex}{2.5ex}\right\}\]

        where $\peqsub{N}{\!H}(x)=\{W\subset H\ /\ \text{open with } x\in W\}$ is the set of open neighbourhoods of $x$ in $H$. The idea behind the definition of $S(U)$ is that if $U$ is ``attached to the boundary of $M$''\footnote{The quotes recall that there exists yet no boundary! However, if we have in mind the examples of the beginning of this section, the quoted ideas work pretty well.}, then $S(U)$ is formed by the geodesics $\gamma=(p,v)$ such that $\gamma$ itself and all close enough geodesics finish their tour on $U$ (and hence close to the ``boundary of $M$''). The next proposition gathers some important properties whose proof can be found in \cite{geroch1968local}:

         \needspace{\baselineskip}
        \begin{properties}\mbox{}
            \begin{enumerate}\renewcommand{\labelenumi}{\thesection.\arabic{enumi}.}
                \setcounter{enumi}{\value{theorem}}
                \item $S(M)=\peqsub{G}{\!I}$
                \item $S(U_1)\cap S(U_2)=S(U_1\cap U_2)$\label{remark interseccion base topology}
                \item If $U$ has compact closure and is sufficiently small, then $S(U)=\emptyset$ i.e.  if $U$ is not ``attached to the boundary'' then the ``end'' of the incomplete geodesics lies outside $U$.\label{S(U)=vacio}
            \setcounter{theorem}{\value{enumi}}
            \end{enumerate}
        \end{properties}

        The first two properties imply that $\beta=\{S(U)\ /\ U \text{ open in }M\}$ is a base of a topology over $\peqsub{G}{\!I}$ (the one formed with all possible unions that we denote $\mathcal{T}_\beta$). The topology $\mathcal{T}_\beta$  allow us to define the following equivalence relation:

        \begin{definitions}\mbox{}\renewcommand{\labelenumi}{\thesection.\arabic{enumi}}
            \begin{enumerate}
                \setcounter{enumi}{\value{theorem}}
                \item Two points $\gamma_1,\gamma_2\in \peqsub{G}{\!I}$ are equivalent ($\gamma_1\sim\gamma_2$) if for every $U_1\in \peqsub{N}{\mathcal{T}_\beta}(\gamma_1)$ we have $\gamma_2\in U_1$ and for every $U_2\in \peqsub{N}{\mathcal{T}_\beta}(\gamma_2)$ we have $\gamma_1\in U_2$.
                \item The set $\partial$ of all equivalence classes $[\gamma]$, with the quotient topology over $\peqsub{G}{\!I}$, is called {\boldmath$g$}\textbf{-boundary}:\vspace*{\dist}

                      \[\partial=\{[\gamma]\ /\ \gamma\in \peqsub{G}{\!I}\}\]

                     We will denote by $\pi$ the quotient map $\pi:\peqsub{G}{\!I}\rightarrow\partial$ such that $\pi(\gamma)=[\gamma]$.
                \item We define the \textbf{completed manifold} $\widehat{M}=M\sqcup \partial$ (where we use $\sqcup$ to remark that $\partial$ is an abstract set formed by equivalence classes, so $A\sqcup B$ will mean $A\subset M$ and $B\subset \partial$).
                \item A subset $U\sqcup\Gamma$ is said to be \textbf{open} in $\widehat{M}$ if $U\subset M$ and $\Gamma\subset \partial$ are open sets in $M$ and $\partial$ respectively, and $\pi^{-1}(\Gamma)\subset S(U)$.
            \setcounter{theorem}{\value{enumi}}
            \end{enumerate}
        \end{definitions}

        The equivalence class relates geodesics that intuitively have the same ``end-point'', and hence $\partial$ is somehow the set of end-points of the incomplete geodesics. The last definition (which indeed defines a base of a topology) tells us how the abstract boundary $\partial$ is attached to the original space $M$, it demands the open set $\Gamma$ to be formed by end-points of incomplete geodesics that get into $U$ and remain there until the end of their parameters.

      \subsection{\texorpdfstring{{\boldmath$g$}}{g}-boundary of the Misner space}
       In this section we will apply the  construction of the $g$-boundary to the Misner space $M^-$ and study how the resulting topology is related with the quotient topology obtained in Section \ref{Topology}. Notice that, up to some signs, we can work on any cone, so to simplify the notation we will consider from now on $M^+$. Actually we are going to work in the Minkowski upper quadrant $Q_4$ with $(t,x)$ coordinates and prove that its  $g$-boundary is, precisely, $\partial Q_4=\{(|x|,x)\ /\ x\in\R\}$. The proof considering the identification goes with slightly change as we will see. The idea of the proof relies on the fact that we are working with some coordinates that cover not only $Q_4$ but the whole $\R^2$. This allows us to ``give explicit coordinates'' to the end points.\vspace*{2ex}

    \begin{minipage}{1\linewidth}
        \begin{wrapfigure}{r}{0.35\textwidth}\vspace*{-3.5ex}
          \centerline{\includegraphics[width=0.35\textwidth]{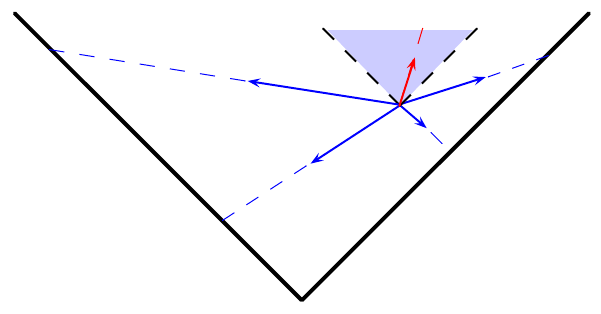}}
        \end{wrapfigure}
        The geodesics of the Minkowski space-time are straight-lines $\peqsub{\gamma}{(p,v)}(\tau)=p+\tau v$ and, in $Q_4$, the incomplete ones are those that hit the semidiagonals $\mathcal{D}=\{(|x|,x)\ /\ x\in\R\}$ when moving forward. The picture on the right shows that a geodesic is incomplete if and only if its velocity vector $v$ points towards $\mathcal{D}$ i.e. $v\notin \overline{Q_4}$ (remember that we consider only what happens for positive values of the parameter). Therefore:\vspace*{\dist}

        \[\peqsub{G}{\!I}=\{(p,v)\in\R^4\ /\ p\in Q_4\text{ and } v\notin\overline{Q_4}\}=Q_4\times\left(\, \overline{Q_4}\, \right)^c\]
    \end{minipage}
\mbox{}\\

        Now for a given initial data $(p,v)=((t,x),(\peqsub{v}{t},\peqsub{v}{x}))\in\peqsub{G}{\!I}$ we want to obtain $\varphi(p,v)$ the length of the parameter and also $\xi(p,v)$, the $x$-coordinate of the hit point. Studying the different possibilities (when it hits the right semidiagonal, the left one or the origin) it is not hard to obtain:\vspace*{\dist}

        \[\begin{array}{ccc}
        \begin{array}{cccc}
            \varphi: & \peqsub{G}{\!I}     & \longrightarrow & \R^+\\
                     & (p,v) &   \longmapsto   & \dfrac{t-\varepsilon x}{\varepsilon \peqsub{v}{x}-\peqsub{v}{t}}
        \end{array} & \qquad\qquad &
        \begin{array}{cccc}
            \xi: & \peqsub{G}{\!I}     & \longrightarrow & \R\\
                 & (p,v) &   \longmapsto   & \dfrac{t\peqsub{v}{x}- x  \peqsub{v}{t}}{\varepsilon \peqsub{v}{x}-\peqsub{v}{t}}
        \end{array}
        \end{array}\]

        where $\varepsilon:=\mathrm{sign}(t \peqsub{v}{x}-x  \peqsub{v}{t})\in\{-1,0,1\}$ tells us where the geodesic hits: $-1$ for the left semidiagonal, $1$ for the right one and $0$ for the origin. Notice that the denominator never vanishes as $v\notin\overline{Q}_4$.\vspace*{2ex}

        At this point we have to compute the sets $S(U)$ for any given open $U$, but it will be enough to focus on the boxes of $Q_4$ of side $2s$ around a point $(t,x)\in\R^2$ i.e. $B((t,x),s):=[I(t,s)\times I(x,s)]\cap Q_4$ where $I(x,\delta)=(x-\delta,x+\delta)$. According to property \refconchap{S(U)=vacio}, if we consider a point $p=(t,x)\notin\mathcal{D}$, then $S(B(p,s))=\emptyset$ for $s$ small enough. For bigger $s$, $S(B(p,s))$ turns out to be the same as if we consider the balls $B(p',s')$ and $B(p'',s'')$ where $p'=(x',x')$ and $p''=(-x'',x'')$ are the projections of $p$ over the semidiagonals $\mathcal{D}$ and $s',s''$ (possibly zero) are given by the intersection of $B(p,s)$ with $\mathcal{D}$. Hence we must focus on points $p=(|x|,x)$ and we will simply denote $B(x,s)=B((|x|,x),s)=I(|x|,s)\times I(x,s)$.\vspace*{2ex}

        The following lemma allows us to control the affine parameter $\varphi$ of two close incomplete geodesics.

        \begin{lemma}\label{lemma varphi-varphi'}\mbox{}\\
           Let $(p,v)\in\peqsub{G}{\!I}$ with associated $\varphi=\varphi(p,v)$ and let us consider another incomplete geodesic $(p',v')\in\peqsub{G}{\!I}$ with $\varphi'=\varphi(p',v')$ such that $|x-x'|<\delta$, $|t-t'|<\delta$, $|v_x-v_x'|<\delta$ and $|v_t-v_t'|<\delta$ for some $\delta>0$ small, then:\vspace*{\dist}

           \[|\varphi'-\varphi|\leq2\delta\frac{1+\varphi}{\left|\varepsilon'v_x-v_t\right|-2\delta}\]
        \end{lemma}
        \begin{proof}\mbox{}\\
           Let us denote $x'=x+\widetilde{x}$ with $|\widetilde{x}|<\delta$ and analogously for the rest of the variables. Expanding the expressions involved, we obtain:\vspace*{\dist}

           \[\varphi'-\varphi=\frac{(\widetilde{t}-\varepsilon'\widetilde{x})-\varphi(\varepsilon'\widetilde{v}_x-\widetilde{v}_t)+(\varepsilon-\varepsilon')\xi}{(\varepsilon'v_x-v_t)+(\varepsilon'\widetilde{v_x}-\widetilde{v_t})}\]

           If $\varepsilon=0$ then $\xi(p,v)=0$, otherwise $\varepsilon=\varepsilon'$ because we are taking $\delta$ small enough such that close geodesics hit the same side. So either way the last term in the numerator is zero. Hence\vspace*{\dist}

           \[|\varphi'-\varphi|=\left|\frac{(\widetilde{t}-\varepsilon'\widetilde{x})-\varphi(\varepsilon'\widetilde{v}_x-\widetilde{v}_t)}{(\varepsilon'v_x-v_t)+(\varepsilon'\widetilde{v_x}-\widetilde{v_t})}\right|\leq2\delta\frac{1+\varphi}{\left|\varepsilon'v_x-v_t\right|-2\delta}\]

           in the inequality we have used the triangle inequality for the numerator and the reverse triangle inequality for the denominator (where we have also used that $\delta$ is small enough such that $|\varepsilon'v_x-v_t|>|\varepsilon'\widetilde{v}_x-\widetilde{v}_t|$).
        \end{proof}

    \begin{minipage}{1\linewidth}
        \begin{wrapfigure}{r}{0.35\textwidth}\vspace*{-3.5ex}
          \centerline{\includegraphics[width=0.35\textwidth]{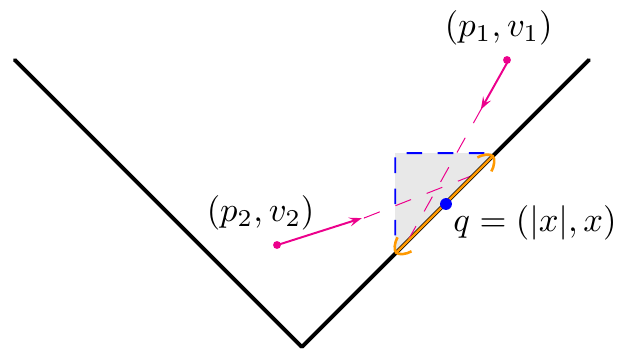}}
        \end{wrapfigure}
        The previous lemma is very important as it relates the affine parameter of two close geodesics, notice that this difference can be taken as small as we want by making $\delta$ small. Now we are going to state a fundamental result that relates the space $Q_4$ with its boundary $\partial Q_4$. It says that the geodesics finishing in an interval of the $\partial Q_4$, are the same that the ones entering and remaining until the end of their parameter into the ball attached to this interval, which geometrically is obvious (but the analytical proof is quite cumbersome). When no confusion is possible, we will omit the variables and write simply $\varphi=\varphi(p,v)$ and $\xi=\xi(p,v)$.\vspace*{2ex}
    \end{minipage}

        \begin{proposition}\label{proposition equal sets}\mbox{}\\
           $S\left(\rule{0ex}{2.4ex} B(\peqsub{x}{0},s)\right)=\xi^{-1}\left(\rule{0ex}{2.2ex}I(\peqsub{x}{0},s)\right)$ where we take $s<|\peqsub{x}{0}|$ if $\peqsub{x}{0}\neq0$.
        \end{proposition}
\begin{proof}\mbox{}\\
           ``$\subset$'' Let $(p,v)=((t,x),(v_t,v_x))\in S(B(\peqsub{x}{0},s))$, then by definition, there exists $U\in \peqsub{N}{\!H}(p,v,\varphi)$ such that $\psi(U\cap H_+)\subset B(\peqsub{x}{0},s)$. As $H=\R^6$ with the usual topology, we may consider by shrinking $U$, that there exist $\delta>0$ and $\nu>0$ such that:\vspace*{\dist}

           \[U=I(t,\delta)\times I(x,\delta)\times I(v_t,\delta)\times I(v_x,\delta)\times I(\varphi,\nu)\]

           Let us now check that $(p,v)\in\xi^{-1}(I(\peqsub{x}{0},s))$, or equivalently that $|\xi(p,v)-\peqsub{x}{0}|<s$:
           \begin{align*}
              |\xi-\peqsub{x}{0}|&=\left|x+v_x(\varphi-\tau)-\peqsub{x}{0}+v_x\tau\right|\leq\\
                            &\leq\left|\rule{0ex}{2.8ex}\pi_x\left(\rule{0ex}{2.2ex}\psi(p,v,\varphi-\tau)\right)-\peqsub{x}{0}\right|+|v_x|\tau< s+|v_x|\tau
           \end{align*}
           where $\pi_x$ is the projection over the $x$ coordinate (remember that $\psi:H_+\rightarrow \R^2$) and $\tau\in(0,\varphi)$ is small to ensure that $(p,v,\varphi-\tau)\in H_+$ and $\psi(p,v,\varphi-\tau)\in B(\peqsub{x}{0},s)$. Finally notice that the last inequality follows from the fact that
           \begin{equation*}
                \psi(U\cap H_+)\subset B(\peqsub{x}{0},s)=[I(|\peqsub{x}{0}|,s)\times I(\peqsub{x}{0},s)]\cap Q_4\quad \longrightarrow \quad \pi_x\left(\rule{0ex}{2.2ex}\psi(U\cap H_+)\right)\subset I(\peqsub{x}{0},s)
           \end{equation*}

            \begin{minipage}{1\linewidth}
        \begin{wrapfigure}{r}{0.34\textwidth}\vspace*{-3ex}
          \centerline{\includegraphics[width=0.34\textwidth]{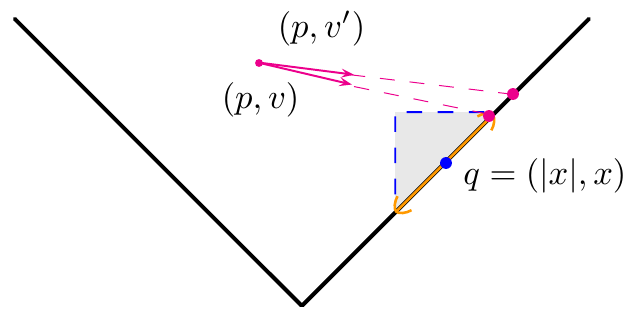}}
        \end{wrapfigure}
        As this inequality holds for every small $\tau$, in particular $|\xi-\peqsub{x}{0}|\leq s$ and it only remains to prove that the equality is not possible. First of all notice that if $\varepsilon(p,v)=0$, then necessarily $\peqsub{x}{0}=0$, as we are taking $s<|\peqsub{x}{0}|$ whenever $\peqsub{x}{0}\neq0$. Hence $|\xi-\peqsub{x}{0}|=|0-0|<s$ and so we are done. We consider now $\varepsilon\neq0$, and let us assume that $|\xi-\peqsub{x}{0}|= s$ which implies that $\xi=\peqsub{x}{0}\pm s$. The idea is shown on the right picture where we see that we have to take a geodesic of $U\cap H_+$ such that it points outside the interval, and we will see that it cannot stay in $B(\peqsub{x}{0},s)$ until the end of its parameter which is a contradiction. Let us consider:
        \end{minipage}

           \[(p,v',\tau')=((t,x),(v_t\pm\varepsilon\eta,v_x\pm\eta),\varphi'-\tau)\in H_+\]

           where $\varphi'=\varphi(p,v')$ and recall that $\xi=\peqsub{x}{0}\pm s$. First of all notice that:
           \begin{align*}
              &\xi(p,v')=\frac{t(v_x\pm\eta)-x(v_t\pm\varepsilon\eta)}{\varepsilon(v_x\pm\eta)-(v_t\pm\varepsilon\eta)}= \frac{tv_x-xv_t\pm\eta(t-x\varepsilon)}{\varepsilon v_x -v_t}=\xi(p,v)\pm\eta\varphi\\
              &\varphi'=\frac{t-\varepsilon x}{\varepsilon(v_x\pm\eta)-(v_t\pm\varepsilon\eta)}=\frac{t-\varepsilon x}{\varepsilon v_x -v_t}=\varphi
           \end{align*}

           The first equation tells us that indeed, taking $\eta$ small, one geodesic hits the semidiagonal close to the other one, while the second one tells that their affine parameter is the same (the change in the velocity is compensated with the change in the space traversed). If we take $0<\eta<\delta$ and $|\tau|<\nu$ then $(p,v',\tau')=(p,v',\varphi-\tau)\in U\cap H_+$ and so $\psi(p,v',\varphi-\tau)\in B(\peqsub{x}{0},s)$ for every valid $\tau>0$. But let us see that then we obtain a contradiction as this new geodesic will eventually get out of $B(\peqsub{x}{0},s)$:
           \begin{align*}
             s&>|x+(v_x\pm\eta)(\varphi-\tau)-\peqsub{x}{0}|=\\
              &=|\pm s\pm\eta\varphi-\tau (v_x\pm\eta)|=|s+\eta\varphi-\tau(\pm v_x+\eta)|
           \end{align*}
           If $(\pm v_x+\eta)\leq0$ then we are done as everything is nonnegative and $s+\eta\varphi+\tau|\pm v_x+\eta|\geq s+\eta\varphi>s$ which is a contradiction. If $(\pm v_x+\eta)>0$ then we take $\tau=\frac{\eta\varphi}{k(\pm v_x+\eta)}$ with $k$ large enough such that $\tau<\nu$, and we obtain again a contradiction.\vspace*{2ex}

           ``$\supset$'' Let $(p,v)=((t,x),(v_t,v_x))\in\xi^{-1}(I(\peqsub{x}{0},s))$, thus $\xi(p,v)=\peqsub{x}{0}+\lambda$ with $|\lambda|<s$. We consider
           \begin{align*}
              U&=I(t,\delta)\times I(x,\delta)\times I(v_t,\delta)\times I(v_x,\delta)\times I(\varphi,\nu)\in \peqsub{N}{\!H}(p,v,\varphi)\\
              \nu&=\frac{1}{m|v_x|}\left(\frac{n-2}{n}s-|\lambda|\right)>0
           \end{align*}

           where $\delta$, $n$ and $m$ will be chosen later on (we take already $n$ big enough such that $\nu>0$) and besides, if $v_x=0$ we consider just $\nu=1$. We have to check now that we can take $\delta$ and $\nu$ small enough such that $\psi(U\cap H_+)\subset B(\peqsub{x}{0},s)$. We consider an arbitrary $(p',v',\varphi+\mu)\in U\cap H_+$, thus $|t-t'|<\delta$ (analogously for the rest of the coordinates), $|\mu|<\nu$ and $0<\varphi+\mu<\varphi'$, the last condition coming from the fact that the point belongs to $H_+$.
           \begin{align*}
              \bigg|\pi_x&\bigg(\rule{0ex}{2ex}\psi(p',v',\varphi+\mu)\bigg)-\peqsub{x}{0}\bigg|=\left|x'+v_x'(\varphi+\mu)-\peqsub{x}{0}\right|=\\
              &=\left|\xi-\peqsub{x}{0}+v_x\mu+(x'-x)+(v_x'-v_x)(\varphi+\mu)\right|<\\
              &<\left|\lambda+v_x\mu\right|+(1+\varphi+\mu)\delta<\left|\lambda+v_x\mu\right|+\frac{s}{n}
           \end{align*}
           The last inequality follows if we take $\delta<\frac{s}{1+\varphi+\mu}\frac{1}{n}$. Now we have to prove that if we take $\delta$ and $\nu$ small enough, then $\left|\lambda+v_x\mu\right|<\frac{n-1}{n}s$ for every $\mu$ such that $|\mu|<\nu$ and $0<\varphi+\mu<\varphi'$, hence for every $\mu\in(-min(\varphi,\nu),min(\varphi'-\varphi,\nu))$. Notice that it is enough to prove this for every $|\mu|<\nu$ (if they are the minima we are done, if they are not we have proved it for a wider range than necessary so we are also done), provided that $\varphi'-\varphi>-\nu$ which is not a problem as the condition $|\varphi-\varphi'|<\nu$ can be achieved according to lemma \ref{lemma varphi-varphi'} taking $\delta$ small enough.\vspace*{2ex}

           Summarizing, it only remains to prove that $\left|\lambda+v_x\mu\right|<\frac{n-1}{n}s$ for every $|\mu|<\nu$. As the function is monotonic in $\mu$ it is in fact enough to prove the inequality for the extrema of this interval:\vspace*{\dist}

           \[\left|\lambda\pm v_x\nu\right|<\frac{n-1}{n}s\]

           If $v_x$ is zero, then the inequality trivially holds considering the definition of $\nu$. Otherwise we have:
           \begin{align*}
              |\lambda\pm& v_x\nu|=\left|sgn(\lambda)|\lambda|\pm v_x\frac{1}{m|v_x|}\left(\frac{n-2}{n}s-|\lambda|\right)\right|=\\
               &=\left| |\lambda|\left(\pm sgn(\lambda)sgn(v_x)-\frac{1}{m}\right) +\frac{n-2}{mn}s\right|\leq\\
               &\leq \left(1+\frac{1}{m}\right)|\lambda| +\frac{n-2}{mn}s\overset{\star}{<}\\
               &=\left(\frac{m+1}{m}+\frac{1}{m}\right)\frac{n-2}{n}s<\frac{n-1}{n}s
           \end{align*}

           where in the $\star$ inequality we have used that $|\lambda|<\frac{n-2}{n}s$ (remember that $\nu>0$) and in the last one we have taken $m$ large enough.
        \end{proof}

        We saw in the previous section that $\beta_1=\{S(U)\ /\ U\in \mathcal{T}\}$ is a base of a topology (where $\mathcal{T}$ denotes the usual topology of $Q_4\subset\R^2$). We can also define $\beta_2=\{S(B(x,s))\ /\ x\in\R,\ s>0\}$. It is clear that for any $(p,v)\in\peqsub{G}{\!I}$ we have that $(p,v)\in\xi^{-1}(I(\xi,s))=S(B(\xi,s))$ and so\vspace*{\dist}

        \[\peqsub{G}{\!I}=\bigcup_{x\in\R}\bigcup_{s>0}S\left(\rule{0ex}{2ex}B(x,s)\right)\]

        This fact, together with property \refconchap{remark interseccion base topology}, implies that $\beta_2$ is also a base of a topology, and we will see in next lemma, that indeed they induce the same topology. Finally let us remark that if we consider the base of $Q_4$ given by $\beta_3=\{B=B((t,x),s)\ /\ (t,x)\in Q_4,\ s<d((t,x),\partial Q_4)\}$, then every $S(B)$ is empty as these open sets are not ``attached to the boundary'' (here $d$ denotes the Euclidean distance).

         \needspace{\baselineskip}
        \begin{lemma}\label{lemma equivalencia}\mbox{}
         \begin{itemize}
           \item $\peqsub{\mathcal{T}}{\beta_2}=\peqsub{\mathcal{T}}{\beta_1}$
           \item Let $\gamma_1,\gamma_2\in \peqsub{G}{\!I}$, then $\gamma_1\sim\gamma_2$ if and only if $\xi(\gamma_1)=\xi(\gamma_2)$.
         \end{itemize}
        \end{lemma}

        \needspace{\baselineskip}
        \begin{proof}\mbox{}
           \begin{itemize}
              \item As $\beta_2\leq\beta_1$ then $\peqsub{T}{\beta_2}\leq\peqsub{T}{\beta_1}$. If we now consider a basic open set $S(U)\in\beta_1$ and $(p,v)\in S(U)$, it is enough to prove that there exists some $S(B)\in\beta_2$ such that $(p,v)\in S(B)\subset S(U)$. Again we have $(p,v)\in\xi^{-1}(I(\xi,s))=S(B(\xi,s))$, let us see that we can take $s$ small enough in such a way that $B(\xi,s)\subset U$, which would imply $(p,v)\in S(B(\xi,s))\subset S(U)$.\vspace*{2ex}

                  Let $V=B(\xi,s)\in E(\xi(p,v))$ be an open basic set of $Q_4$. As $(p,v)\in S(U)$ then $\psi(p,v,\varphi-\tau)\in U$ for every $\tau$ small enough, on the other hand $(p,v)\in S(V)$ and so we have also $\psi(p,v,\varphi-\tau)\in V$ for every $\tau$ small enough making $U\cap V\neq\phi$ for every $V$, and hence $\xi(p,v)\in\overline{U}$. $U\subset Q_4$ and $\xi(p,v)\in\partial Q_4$. So as we expected, $U$ is attached to the boundary, but apparently it could happen that $\overline{U}\cap\partial Q_4=\{(|\xi(p,v)|,\xi(p,v))\}$. However by applying the same argument of proposition \ref{proposition equal sets}, we can prove that $|\xi(p,v)-\peqsub{x}{0}|=s$ is not possible, hence we conclude that there should exist a whole interval $I(\xi(p,v),\peqsub{s}{0})$ in $\overline{U}\cap\partial Q_4$ and so $B(\xi,\peqsub{s}{0})\subset U$.
              \item The left implication is clear according to the definition. For the other implication let us suppose that $r=\xi(\beta)-\xi(\alpha)>0$, then we consider the open set $U=\xi^{-1}(I(\xi(\alpha),r/2))\in E(\alpha)$. As $|\xi(\alpha)-\xi(\beta)|=r>r/2$ then $\beta\notin U$ and therefore $\beta\nsim\alpha$.\vspace*{-4ex}
           \end{itemize}
        \end{proof}

        Notice that if we consider $M^+$ (i.e. $Q_4$ with the identification provided by $\peqsub{G}{\theta_0}$), then proposition \ref{proposition equal sets} is still valid as it is a set equality. In particular, saturating (i.e. considering the union), we have:\vspace*{\dist}

        \[S\left(\rule{0ex}{2.2ex} S_{\sim}[B(\peqsub{x}{0},s)]\right)=\xi^{-1}\left(\rule{0ex}{2.2ex}S_{\sim}[I(\peqsub{x}{0},s)]\right)\]

        Thus the first point of lemma \ref{lemma equivalencia} is also valid in $M^+$. The second one becomes $\gamma_1\sim\gamma_2$ if and only if\vspace*{\dist}

        \[\Big(\mathrm{sign}(\xi(\alpha))=\mathrm{sign}(\xi(\gamma_2))=0\Big)\vee\Big(\mathrm{sign}(\xi(\alpha))=\mathrm{sign}(\xi(\gamma_2))\neq 0 \wedge  \ln|\xi(\gamma_1)|\equiv\ln|\xi(\gamma_2)|\ \mathrm{mod}\, \peqsub{\theta}{0}\Big)\]

        as a consequence of the exponential behaviour of the group action over the semidiagonals.

        \begin{lemma}\mbox{}\\
          The $g$-boundary $\partial$ of $Q_4$ is homeomorphic to $\R$.\label{lemma partia cong R}
        \end{lemma}
        \begin{proof}\mbox{}\\
           $\xi:\peqsub{G}{\!I}\rightarrow \R$ is surjective as  the geodesics fill the whole plane, continuous as can be seen from the explicit expression over $\peqsub{G}{\!I}=Q_4\times\left(\,Q_4^c\,\right)$, and open as the gradient $\nabla\xi$ does not vanish anywhere in $\peqsub{G}{\!I}$ (it is then a submersion). These three properties imply that $\xi$ is a quotient map. Hence\vspace*{\dist}

           \[\fracdiag{\peqsub{G}{\!I}}{\sim_\xi}\cong\R\]
           where $\sim_f$ is the equivalence relation that relates two points $x,y$ if $f(x)=f(y)$. Finally notice that $\sim_\xi$ is precisely $\sim$ according to the previous lemma, thus:

           \[\partial\overset{def}{=}\fracdiag{\peqsub{G}{\!I}}{\sim}=\fracdiag{\peqsub{G}{\!I}}{\sim_\xi}\cong\R\]

           which completes the proof of the lemma.
        \end{proof}

        In order to take into account the boost identification, we have to define the map $\widetilde{\xi}:\peqsub{G}{I}\rightarrow \S_+^1\sqcup\S_-^1\sqcup\{\star\}$ (where $\{\star\}$ is a one set point) mapping $(p,v)\mapsto exp(i\ln(\pm\xi)+i\frac{2\pi}{\theta_0})\in\S^1_\pm$ when $sgn(\xi(p,v))=\pm1$ and $(p,v)\mapsto \star$ if $\xi(p,v)=0$. It is open as it is the composition of $\xi$ with the projection which is open. Considering this quotient map we see that $\partial\cong \S_+^1\sqcup\S_-^1\sqcup\{\star\}$.\vspace*{2ex}

        We now end with the following theorem that connects the results obtained in the first and second part of the paper.

        \begin{theorem}\mbox{}\\
           $\widehat{Q_4}$ is homeomorphic to $\overline{Q}_4$ with the usual topology.
        \end{theorem}
        \begin{proof}\mbox{}\\
         Let us recall that $\widehat{Q_4}=Q_4\sqcup \partial$ and its topology $T$ is given by: $U\sqcup \Gamma$ is open over $\widehat{Q_4}$ if $U\subset Q_4$ and $\Gamma\subset\partial$ are open sets and $\pi^{-1}(\Gamma)\subset S(U)$. Now we define:

         \[\begin{array}{cccc}
            \Phi : & \widehat{Q_4}              & \longrightarrow & \left(\, \overline{Q}_4,T_{usual}\, \right)\\
                   &  x\in Q_4                  &   \longmapsto   & x\in Q_4 \\
                   &  \alpha\in\partial &   \longmapsto   & (|\xi(p,v)|,\xi(p,v))\in \partial Q_4
        \end{array}\]

        where $\alpha=[(p,v)]$. $\Phi$ is well defined as $\xi(p,v)$ is the same for every representative $(p,v)$ of $\alpha$ and so it is bijective. Let us prove that it is a homeomorphism.\vspace*{2ex}

        Let $U$ be an open set contained in $Q_4$, then $\Phi^{-1}(U)=U\sqcup\phi$ is open as they are both open sets such that $\pi^{-1}(\phi)=\phi\subset S(U)$. If $U$ is not entirely contained in $Q_4$ then we may assume that it is of the form $U=B(x,s)\sqcup\Delta(I(x,s))$ where $\Delta:\R\rightarrow \partial Q_4$ is given by $\Delta(x)=(|x|,x)$. Notice that $\Phi([(p,v)])=\Delta(\xi(p,v))$ over $\partial$. Hence\vspace*{\dist}

        \[\Phi^{-1}(U)=B(x,s)\sqcup \pi(\xi^{-1}(I(x,s)))\]

        that is an open set as they are both open sets in their respective spaces, and they also satisfy $\pi^{-1}(\pi(\xi^{-1}(I(x,s))))=\xi^{-1}(I(x,s))=S(B(x,s))$. Notice that the first equality, which is not true in general, holds in this particular case as $\xi^{-1}(A)$ is always saturated by lemma \ref{lemma equivalencia}. So $\Phi$ is continuous.\vspace*{2ex}

        Let us check that it is also open and hence an homeomorphism. Let $U\sqcup\Gamma$ be an open set of $\widehat{Q_4}$, then $U$ and $\Gamma$ are open sets of $Q_4$ and $\partial$ respectively and $\pi^{-1}(\Gamma)\subset S(U)$.\vspace*{\dist}

        \[\Phi(U\sqcup\Gamma)=U\sqcup (\Delta\circ\xi\circ\pi^{-1})(\Gamma)\]

        If we take some $(t,x)\in U$, as it is open in $Q_4$ then we can find a neighbourhood of $(t,x)$ contained in $U$ and hence in $\Phi(U\sqcup\Gamma)$. If we consider $p=(|x|,x)\in(\Delta\circ\xi\circ\pi^{-1})(\Gamma)$, then there exists $I=I(\peqsub{x}{0},\peqsub{s}{0})$ such that $p\in \Delta(I)\subset (\Delta\circ\xi\circ\pi^{-1})(\Gamma)$ as it is open in $\partial Q_4$ ($\pi$ is continuous, $\xi$ is open and $\Delta$ is a homeomorphism) and so $p\in U\sqcup\Delta(I)\subset\Phi(U\sqcup\Gamma)$. If we prove that $U\sqcup\Delta(I)$ is open over $\overline{Q}_4$ we would conclude that $\Phi(U\sqcup\Gamma)$ is also open. We have for every element $(|x'|,x')\in\Delta(I)\subset\Delta\circ\xi\circ\pi^{-1}(\Gamma)\subset \Delta\circ\xi(S(U))$, then there exists some $(p,v)\in S(U)$ such that $\xi(p,v)=x'$ and by definition there exists $A\in \peqsub{N}{\!H}(p,v,\varphi)$ such that $\psi(A\cap H_+)\subset U$. In particular there are some elements of $A$ of the form\vspace*{\dist}

        \[(p,v',\tau')=((t,x),(v_t\pm\varepsilon\eta,v_x\pm\eta),\varphi'-\tau)\in H_+\]

        with $\eta>0$. As we proved in proposition \ref{proposition equal sets}, $\varphi'=\varphi$ and $\xi(p,v')=\xi(p,v)\pm\eta\varphi$ so if $\eta<\frac{\peqsub{s}{0}-|\xi-\peqsub{x}{0}|}{\varphi}$ then $\xi(p,v')\in\Delta(I)$ for every $\eta>0$. The images of these geodesics form a cone with apex $p\in Q_4$, and all the geodesics enter $U$ and remain there until the end of their parameter ($\varphi$ for all the geodesics), so they all finish over $\Delta(I)$. Hence taking the image of these geodesics inside $U$ and $\Delta(I)$ we obtain a truncated open cone which is an open neighbourhood of $(|x'|,x)$ in $U\sqcup\Delta(I)$, so this last set is also open and hence $\Phi$ is an open map.
      \end{proof}

        Again it follows the analog result when $M^+=Q_4/\peqsub{G}{\theta_0}$ is considered, where the $\Phi$ function is defined similarly, leaving the interior points unchanged and the points of $\tilde{\partial}$ are mapped via the function $\widetilde{\xi}$. This completes the proof of the fact that the $g$-boundary recovers the boundary and the topology obtained when we consider the quotient of a closed quadrant under the group generated by a discrete boost.

      \section{Conclusions}\label{conclusions}
         In this paper we have explicitly obtained the quotient topology of the complete Misner space $\R^2_1/boost$. We find a T0 but not T1 space that is not smooth at the origin, because it is a fixed point under the action of Lorentz boosts. When the origin is removed a T1 but not T2 smooth space is obtained and finally, when just half plane over/under a diagonal is considered, we obtain a T2 smooth manifold. The behaviour of the geodesics with respect to the four circles (obtained by making the identifications over the four open semi-diagonals) strongly resembles to the behaviour of a $g$-boundary, so we have computed the $g$-boundary of the Misner space and its associated topology. We have found that indeed there is a natural identification of the $g$-boundary $\partial$ of a cone $\widetilde{Q}_i$, with the circles and the origin of the complete Misner space, and that the topology of $\widehat{Q}_4$ is the same as the (quotient) topology of $\widetilde{Q}_4\cup\widetilde{D}_3\cup\widetilde{D}_4\cup\{0\}$.

      \section*{Acknowledgments}
        The authors are very grateful to Juan Margalef Roig and Miguel S\'anchez Caja for their useful comments and support, and specially to Fernando Barbero and Robert Geroch for their patience, comments and priceless help. This work has been supported by the Spanish MINECO research grant  FIS2012-34379 and the Consolider-Ingenio 2010 Program CPAN (CSD2007-00042).

{\scriptsize
        \bibliographystyle{amsplain}      
        \bibliography{biblioteca}\addcontentsline{toc}{section}{References}

\providecommand{\bysame}{\leavevmode\hbox to3em{\hrulefill}\thinspace}
\providecommand{\MR}{\relax\ifhmode\unskip\space\fi MR }
\providecommand{\MRhref}[2]{%
  \href{http://www.ams.org/mathscinet-getitem?mr=#1}{#2}
}
\providecommand{\href}[2]{#2}
\begin{thebibliography}{10}

\bibitem{durin2006closed}
B.~Durin Bruno and B.~Pioline, \emph{Closed strings in misner space: A toy
  model for a big bounce?},
  \href{http://link.springer.com/chapter/10.1007/1-4020-3733-3_8}{String
  Theory: From Gauge Interactions to Cosmology}
  [\href{http://arxiv.org/pdf/hep-th/0501145.pdf}{arXiv:hep-th/0501145v2}],
  Springer, 2006, pp.~177--200.

\bibitem{Hausdorff_separability}
J.L. Flores, J.~Herrera, and M.~S\'{a}nchez, \emph{Hausdorff separability of
  the boundaries for spacetimes and sequential spaces}, Preprint (2014).

\bibitem{geroch1968local}
R.~Geroch, \emph{Local characterization of singularities in general
  relativity},
  \href{http://scitation.aip.org/content/aip/journal/jmp/9/3/10.1063/1.1664599}{Journal
  of Mathematical Physics} \textbf{9} (1968), 450.

\bibitem{geroch1968singularity}
\bysame, \emph{What is a singularity in general relativity?},
  \href{https://www.sciencedirect.com/science/article/pii/0003491668901449}{Annals
  of Physics} \textbf{48} (1968), no.~3, 526--540.

\bibitem{geroch1982singular}
R.~Geroch, L.~Can-bin, and R.M. Wald, \emph{Singular boundaries of
  space--times},
  \href{http://scitation.aip.org/content/aip/journal/jmp/23/3/10.1063/1.525365}{Journal
  of Mathematical Physics} \textbf{23} (1982), 432.

\bibitem{hajicek1970embedding}
P.~Hajicek, \emph{Embedding of singularities},
  \href{http://link.springer.com/article/10.1007/BF00759200}{General Relativity
  and Gravitation} \textbf{1} (1970), no.~1, 27--29.

\bibitem{Sing_geom_Haw_1966}
S.W. Hawking, \emph{Singularities and the geometry of space-time}, Unpublished
  essay submitted for the Adams Prize, Cambridge University (1966).

\bibitem{Haw_y_ellis}
S.W. Hawking and G.F.R. Ellis, \emph{The large scale structure of space-time},
  \href{http://www.cambridge.org/us/academic/subjects/physics/cosmology-relativity-and-gravitation/large-scale-structure-space-time}{Cambrigde
  University Press}, 1973.

\bibitem{hikida2005d}
Y.~Hikida, R.R. Nayak, and K.L. Panigrahi, \emph{D-branes in a big bang/big
  crunch universe: Misner space},
  \href{http://iopscience.iop.org/1126-6708/2005/09/023}{Journal of High Energy
  Physics}
  [\href{http://arxiv.org/pdf/hep-th/0508003v2.pdf}{arXiv:hep-th/0508003v2}]
  \textbf{2005} (2005), no.~09, 023.

\bibitem{JAVALOYES_SANCHEZ_lorentz}
M.A. {Javaloyes Victoria} and M.~{S{\'a}nchez Caja},
  \emph{\href{http://digap.ugr.es/produccion-cientifica/libros/ver_detalles/151292/}{An
  introduction to Lorentzian Geometry and its applications}}, Ed. Universidad
  de Sao Paulo, 2010.

\bibitem{jonsson2005visualizing}
R.M. Jonsson, \emph{Visualizing curved spacetime},
  \href{http://scitation.aip.org/content/aapt/journal/ajp/73/3/10.1119/1.1830500}{American
  journal of physics}
  [\href{http://arxiv.org/pdf/0708.2483.pdf}{arXiv:0708.2483v1}] \textbf{73}
  (2005), 248.

\bibitem{levy1980speed}
J.M. L{\'e}vy-Leblond, \emph{Speed(s)},
  \href{http://scitation.aip.org/content/aapt/journal/ajp/48/5/10.1119/1.12093}{Am.
  J. Phys.} \textbf{48} (1980), 345--347.

\bibitem{roig1993introduccion}
J.~{Margalef Roig} and E.~{Outerelo Dom{\'\i}nguez}, \emph{Introducci{\'o}n a
  la topolog{\'\i}a},
  \href{http://books.google.es/books/about/Introducci%C3%B3n_a_la_topolog%C3%ADa.html?id=GCqFsSyvDo4C}{Ed.
  Complutense}, 1993.

\bibitem{Taub-nut_contraejemplo}
C.W. Misner,
  \emph{\href{http://adsabs.harvard.edu/abs/1967rta1.book..160M}{Taub-{\uppercase{NUT}}
  as a counterexample to almost anything}}, Technical Report, University of
  Maryland \textbf{529} (1965), 1--22.

\bibitem{munkres2000topology}
J.R. Munkres, \emph{Topology}, vol.~2,
  \href{http://books.google.com.br/books/about/Topology.html?id=XjoZAQAAIAAJ\&redir_esc=y}{Prentice
  Hall Upper Saddle River}, 2000.

\bibitem{Oneill_Semirimannian}
B.~O'Neill, \emph{Semi-{\uppercase{r}iemannian} geometry: with applications to
  relativity},
  \href{http://books.google.com.br/books/about/Semi_Riemannian_Geometry_With_Applicatio.html?id=CGk1eRSjFIIC\&redir_esc=y}{Ac.
  Press}, 1983.

\bibitem{thorne1993misner}
K.S. Thorne, \emph{Misner space as a prototype for almost any pathology},
  Directions in General Relativity: Papers in Honor of Charles Misner,
  \href{http://ebooks.cambridge.org/chapter.jsf?bid=CBO9780511628863\&cid=CBO9780511628863A036}{Volume
  1}, vol.~1, 1993, p.~333.

\bibitem{willard2004general}
S.~Willard, \emph{General topology},
  \href{http://books.google.com.br/books/about/General_Topology.html?id=-o8xJQ7Ag2cC\&redir_esc=y}{Courier
  Dover Publications}, 2004.

\end{thebibliography}
}

\clearpage
        \appendix \phantomsection\addcontentsline{toc}{section}{\appendixname}

        \clearpage

        \appendixpage

   \section{Some Topological Results}\label{section math back}
     \subsection{Separation axioms}\label{subsection hausdorff}
      For a given topological space $(M,\mathcal{T})$ and every $p\in M$ we denote $N(p)=\{U\subset M\ /\ \text{open s.t. }p\in U\}$ the set of all open neighbourhoods of $p$. Sometimes to emphasize we will write $\peqsub{N}{\mathcal{T}}(p)$ or, if the topology is obvious from the context, $\peqsub{N}{M}(p)$.

          \begin{definitions}\mbox{}\renewcommand{\labelenumi}{\thesection.\arabic{enumi}}\\
            Let $M$ be a topological space, we say that
              \begin{enumerate}
                  \setcounter{enumi}{\value{theorem}}
                  \item $M$ is \textbf{T2} (or \textbf{Hausdorff}) if for every different $x,y\in M$ there exist $U_1\in N(x)$ and $U_2\in N(y)$ such that $U_1\cap U_2=\emptyset$.
                  \item $M$ is \textbf{T1} if for every different $x,y\in M$ there exist $U_1\in N(x)$ and $U_2\in N(y)$ such that~$\left\{\begin{array}{l}
                      \hspace*{-1ex}x\notin U_2\\
                      \hspace*{-1ex}y\notin U_1
                      \end{array}\right.$
                  \item $M$ is \textbf{T0} if given two distinct points $x,y\in M$ there exists $U_1\in N(x)$ such that $y\notin U_1$ OR there exists $U_2\in N(y)$ such that $x\notin U_2$.
              \setcounter{theorem}{\value{enumi}}
              \end{enumerate}
          \end{definitions}

          \begin{remarks}\mbox{}\renewcommand{\labelenumi}{\thesection.\arabic{enumi}}
              \begin{enumerate}
                  \setcounter{enumi}{\value{theorem}}
                  \item T2 $\ \Rightarrow\ $  T1 $\ \Rightarrow\ $ T0
                  \item A space $M$ is T1 if and only if every one-point set $\{p\}$ is closed.
                  \item A space $M$ is not T0 if and only if two points have exactly the same neighbourhoods.
              \setcounter{theorem}{\value{enumi}}
              \end{enumerate}
          \end{remarks}

      \subsection{Some results on quotient topologies}
        \begin{definition}\label{def quotient topology}\mbox{}\\
            Given an equivalence relation $\sim$ over $M$, we define the \textbf{quotient topology} as the finest topology $\mathcal{T}'$ over $\quot{M}$ such that $p:(M,\mathcal{T})\rightarrow(\quot{M},\mathcal{T}')$ is continuous. Such topology is denoted as $\quot{T}$.
        \end{definition}

        In order to work with this topology, it is useful to introduce a more explicit characterization that can be found in almost any book of general topology \cite{willard2004general,munkres2000topology,roig1993introduccion}, but first we need some definitions:

        \needspace{\baselineskip}
        \begin{definitions}\mbox{}\renewcommand{\labelenumi}{\thesection.\arabic{enumi}}
            \begin{enumerate}
                \setcounter{enumi}{\value{theorem}}
                \item We call the \textbf{saturation} of $U\subset M$ to $S_\sim[U]\equiv p^{-1}(p(U))$ where $p$ is the natural projection.
                \item We say that a subset $U\subset M$ is \textbf{saturated} if $S_\sim[U]=U$.
            \setcounter{theorem}{\value{enumi}}
            \end{enumerate}
        \end{definitions}

        \begin{lemma}\label{lema saturados}\mbox{}
           \begin{itemize}
              \item $p(p^{-1}(B))=B$ for every $B\subset \quot{M}$\\[-1ex]
              \item $\quot{\mathcal{T}}=\{p(H)\ /\ H\in \peqsub{\mathcal{T}}{M} \text{ and it is saturated}\}$
           \end{itemize}
        \end{lemma}

        One might expect that the saturated open sets can be obtained by saturating all the open sets, unfortunately we will obtain in general some saturated subsets that are not open. For instance, the set $J=[0,1/2)$ is open in $I=[0,1]$, but if we identify the endpoints of $I$ and saturate $J$, we obtain $S_\sim J=[0,1/2)\cup\{1\}$ which is not open in $I$. Fortunately the equivalence relations that we use in that paper are quite particular in this respect and this problem will not arise.

        \begin{definition}  \mbox{}\\
            Given a homeomorphism $f:M\rightarrow M$, we define the \textbf{{\boldmath$f$}-equivalence relation} as:\vspace*{\dist}

                   \[x\sim y\ \text{ if and only if  there exists some }n\in\Z\ \ /\ \ y=f^{\left.n\right)}(x)\]

                   which can be summarized by saying that $x\sim f^{\left.n\right)}(x)$ for every $n\in\Z$.
        \end{definition}

        $G=\langle f\rangle=\{f^{\left.n\right)}\ /\ n\in\Z\}$ is a cyclic subgroup of the group $\mathrm{Hom}(M)$ of homeomorphisms of $M$.

        \needspace{\baselineskip}
        \begin{remark}\mbox{}\\
         The saturation of any open set $U$ is always open as can be seen using the following identity:\vspace*{\dist}

                    \[S_\sim=\bigcup_{n=-\infty}^\infty f^{\left.n\right)}\]

                    where both sides have to be thought as operators acting on the subsets of $M$. \label{remark saturate f}
        \end{remark}

        Finally we can characterize the quotient topology and a base of it in a suitable way for our purposes. As it is essential for the paper and we have not found any proof of this characterization (for this particular case), we provide a proof in the following lemma.

        \begin{lemma}\mbox{}\\
            Let $f:M\rightarrow M$ be an homeomorphism and $\sim_f$ its equivalence relation (see proof of lemma \ref{lemma partia cong R}), then:
            \begin{itemize}
               \item $\displaystyle\quotf{\mathcal{T}}=\left\{\rule{0ex}{2.5ex}p\bigl(S_{\sim_f}[G]\bigr)\ /\ G\in \mathcal{T}\right\}$
               \item If $\beta$ is a base of $\mathcal{T}$, then $\peqsub{\beta}{\sim_f}=\left\{\rule{0ex}{2.5ex}p\bigl(S_{\sim_f}[B]\bigl)\ /\ B\in \beta\right\}$ is a base of $\quotf{\mathcal{T}}$.\label{recipe of quotient topology}
            \end{itemize}
        \end{lemma}
        \textbf{Proof}\mbox{}
           \begin{itemize}
              \item ``$\supset$'' Let us denote $\mathcal{T}'=\{p(S_{\sim_f}[G])\ /\ G\in \mathcal{T}\}$, and let $U\in \mathcal{T}'$. Then there exists some $G\in \mathcal{T}$ such that $U=p(S_{\sim_f}[G])$. As we have seen on remark \ref{remark saturate f}:\vspace*{\dist}

                  \[H\equiv S_{\sim_f}[G]=\bigcup_{n\in\Z}f^{\left.n\right)}(G)\]

                  which is open, as $f$ is a homeomorphism, and saturated by definition. Therefore $U=p(H)$ for some $H\in \mathcal{T}$ saturated, so $U\in\quotf{\mathcal{T}}$ according to the second point of lemma \ref{lema saturados}.\\[2ex]
                  ``$\subset$'' Let $U\in\quotf{\mathcal{T}}$, then $U=p(H)$ for some $H\in \mathcal{T}$ saturated, as it is saturated, then by definition we have $H=S_{\sim_f}[H]$. Thus taking $G=H$ we have $U=p(G)=p(S_{\sim_f}[G])$ for some $G\in \mathcal{T}$ and therefore $U\in \mathcal{T}'$.
              \item First notice that the fact that $\beta\subset \mathcal{T}$ implies by the previous point that $\beta_{\sim_f}\subset\quotf{\mathcal{T}}$ as it is required to be a base. Now let us consider a generic open set $U$ of $\quotf{\mathcal{T}}$, then again by the previous point there exists some $G\in \mathcal{T}$ such that $U=p(S_{\sim_f}[G])$. As $\beta$ is a base of $\mathcal{T}$, we have $G=\bigcup B_i$ for some $B_i\in\beta$, thus:\vspace*{-2ex}

                  \begin{align*}
                     U&=p(S_{\sim_f}[G])=p\left(S_{\sim_f}\left[\bigcup_{i\in I} B_i\right]\right)\overset{\ref{remark saturate f}}{=}\\
                      &=p\left(\bigcup_{n\in\Z}f^{\left.n\right)}\left[\bigcup_{i\in I} B_i\right]\right)\overset{\star}{=}p\left(\bigcup_{n\in\Z}\bigcup_{i\in I} f^{\left.n\right)}(B_i)\right)=\\
                      &=p\left(\bigcup_{i\in I} \bigcup_{n\in\Z}f^{\left.n\right)}(B_i)\right)\overset{\ref{remark saturate f}}{=}p\left(\bigcup_{i\in I} S_{\sim_f}\left[B_i\right]\right)\overset{\star}{=}\bigcup_{i\in I} p\left(S_\sim\left[B_i\right]\right)
                  \end{align*}
                  where on the $\star$ equalities we have used $F(\bigcup A_i)=\bigcup F(A_i)$ for arbitrary unions.\hfill\QED
           \end{itemize}

        So finally we have reached a very convenient way to handle the quotient topology in our particular case: we need to consider a local base for every $x\in M$, saturate those neighbourhoods, and project them through $p:M\rightarrow \quotf{M}$ to obtain a base of the quotient topology, which is enough to describe the whole topology. We will make extensive use of this idea in the paper.

        \section{Actions}\label{subsection actions}
        We consider an action given by homeomorphisms, which introduces naturally an equivalence relation identifying the points with its images under the action. The quotient defined by this equivalence relation is usually denoted by $\fracdiag{M}{\ G}\equiv\quot{M}$. Our aim now is to determine what conditions have to be fulfilled by the action and the space in order to obtain a manifold. As we are just interested in $G\subset \mathrm{Hom}(M)$ (specifically $G=\{f^{\left.n\right)}\ /\ n\in\Z\}\subset \mathrm{Hom}(\R^2)$), we can forget about continuity issues.

          \needspace{\baselineskip}
          \begin{definitions}\mbox{}\renewcommand{\labelenumi}{\thesection.\arabic{enumi}}
              \begin{enumerate}
                  \setcounter{enumi}{\value{theorem}}
                  \item An action is \textbf{free} if for every $x\in M$ we have $g(x)\neq x$ for all $g\in G\setminus\{e\}$. Free means that $g$ has no fixed points if $g\neq e$ i.e. it moves all the points.
                  \item An action is \textbf{properly discontinuous} if:\label{definition properly discont}
                   \begin{itemize}
                      \item[$PD1$:] For every $x\in M$ there exists some neighbourhood $U\in N(x)$ such that $g(U)\cap U=\emptyset$ for all $g\in G$ satisfying $g(x)\neq x$.
                      \item[$PD2$:] If we have $x,y\in M$ not in the same orbit ($x\notin\mathcal{O}(y)$), there exist neighbourhoods $U\in N(x)$ and $V\in N(y)$ such that $g(U)\cap V=\emptyset$ for every $g\in G$.
                   \end{itemize}
                   Intuitively, $PD1$ means that if $g$ moves points, then it moves also their neighbourhoods. $PD2$ means that points from different orbits can be separated.
              \setcounter{theorem}{\value{enumi}}
              \end{enumerate}
          \end{definitions}

          On \cite[chap.7]{Oneill_Semirimannian} it is proven that given a manifold $M$ and a subgroup $G\subset \mathrm{Diff}(M)$ of diffeomorphisms, if the action of $G$ on $M$ is free and properly discontinuous, then $\fracdiag{M}{\ G}$ is a Hausdorff manifold. From the proof it follows that the condition $PD2$ implies Hausdorff, hence if the action is free but only verifies $PD1$, we obtain a smooth manifold, but not necessarily Hausdorff (see also \cite[5.8]{Haw_y_ellis}).

\end{document}